\documentclass[a4paper,UKenglish,cleveref, autoref, thm-restate]{lipics-v2021}
\usepackage{algorithm,algorithmic}
\usepackage{thm-restate}
\usepackage{todonotes}
\hideLIPIcs
\title{Approximating Pareto Sum via Bounded Monotone Min-Plus Convolution} %TODO Please add

\author{Geri Gokaj}{Karlsruhe Institute of Technology }{geri.gokaj@kit.edu}{https://orcid.org/0009-0002-7500-6848}{Partially supported by the Deutsche Forschungsgemeinschaft (DFG, German Research Foundation) - 462679611.}
\author{Marvin Künnemann}{Karlsruhe Institute of Technology }{marvin.kuennemann@kit.edu}{}{Partially supported by the Deutsche Forschungsgemeinschaft (DFG, German Research Foundation) - 462679611.}
\author{Sabine Storandt}{University of Konstanz}{sabine.storandt@uni-konstanz.de}{https://orcid.org/0000-0001-5411-3834}{}
\author{Carina Truschel}{University of Konstanz}{carina.truschel@uni-konstanz.de}{https://orcid.org/0009-0009-7582-7209}{Funded by the Deutsche Forschungsgemeinschaft (DFG, German Research Foundation) – Project-ID 251654672 – TRR 161.}
\authorrunning{Gokaj, Künnemann, Storandt, Truschel} %TODO mandatory. First: Use abbreviated first/middle names. Second (only in severe cases): Use first author plus 'et al.'

\Copyright{Jane Open Access and Joan R. Public} %TODO mandatory, please use full first names. LIPIcs license is "CC-BY";  http://creativecommons.org/licenses/by/3.0/

 %TODO mandatory: Please choose ACM 2012 classifications from https://dl.acm.org/ccs/ccs_flat.cfm 

 %TODO mandatory; please add comma-separated list of keywords

\category{} %optional, e.g. invited paper

\relatedversion{} %optional, e.g. full version hosted on arXiv, HAL, or other respository/website
\supplement{}

\supplementdetails[subcategory={Source Code}, cite={}, swhid={swh:1:dir:53715b7d04ec81c190b784371b4675bc28183aee}]{Software}{https://github.com/ggokaj/SoCG26-ParetoSum}

\acknowledgements{}%optional

\nolinenumbers %uncomment to disable line numbering

\newcommand{\rep}{\mathrm{rep}}

\EventEditors{John Q. Open and Joan R. Access}
\EventNoEds{2}
\EventLongTitle{42nd Conference on Very Important Topics (CVIT 2016)}
\EventShortTitle{}
\EventAcronym{CVIT}
\EventYear{2016}
\EventDate{December 24--27, 2016}
\EventLocation{Little Whinging, United Kingdom}
\EventLogo{}
\SeriesVolume{42}
\ArticleNo{23}
\begin{document}

\maketitle

\begin{abstract}
The Pareto sum of two-dimensional point sets $P$ and $Q$ in $\mathbb{R}^2$ is defined as the skyline of the points in their Minkowski sum. The problem of efficiently computing the Pareto sum arises frequently in bi-criteria optimization algorithms. Prior work establishes that computing the Pareto sum of sets $P$ and $Q$ of size $n$ suffers from conditional lower bounds that rule out strongly subquadratic $O(n^{2-\epsilon})$-time algorithms, even when the output size is $\Theta(n)$. Naturally, we ask: How efficiently can we \emph{approximate} Pareto sums, both in theory and practice? Can we beat the near-quadratic-time state of the art for exact algorithms?

On the theoretical side, we formulate a notion of additively approximate Pareto sets and show that computing an approximate Pareto set is \emph{fine-grained equivalent} to Bounded Monotone Min-Plus Convolution. Leveraging a remarkable $\tilde{O}(n^{1.5})$-time algorithm for the latter problem (Chi, Duan, Xie, Zhang; STOC '22), we thus obtain a strongly subquadratic (and conditionally optimal) approximation algorithm for computing Pareto sums.

On the practical side, we engineer different algorithmic approaches for approximating Pareto sets on realistic instances. Our implementations enable a granular trade-off between  approximation quality and running time/output size  compared to the state of the art for exact algorithms established in (Funke, Hespe, Sanders, Storandt, Truschel; Algorithmica '25). Perhaps surprisingly, the (theoretical) connection to Bounded Monotone Min-Plus Convolution remains beneficial even for our implementations: in particular, we implement a simplified, yet still subquadratic version of an algorithm due to Chi, Duan, Xie and Zhang, which on some sufficiently large instances outperforms the competing quadratic-time approaches.	
\end{abstract}

\section{Introduction}In multi-criteria decision-making processes and in algorithms for multi-objective optimization, a central computational task is to obtain and filter non-dominated solutions. This is essential in order to keep intermediate solution sets interpretable and computationally tractable. Often, partial solutions must be combined in the course of the algorithm, and non-dominated filtering must be applied to the resulting set. This core problem is known as \emph{Pareto sum}: Given two sets $P, Q \subset \mathbb{R}^2$ of non-dominated points\footnote{We call a set $S$ of points non-dominated, if there do not exist distinct vectors $s,s' \in S$ such that $s\leq s'$.} (also called skylines or Pareto sets), the Pareto sum $PS(P,Q)$ is defined as the skyline of their Minkowski sum $P + Q := \{p+q \mid ~p \in P, q \in Q\}$.

\begin{figure}
    \centering
    \includegraphics[width=0.48\linewidth]{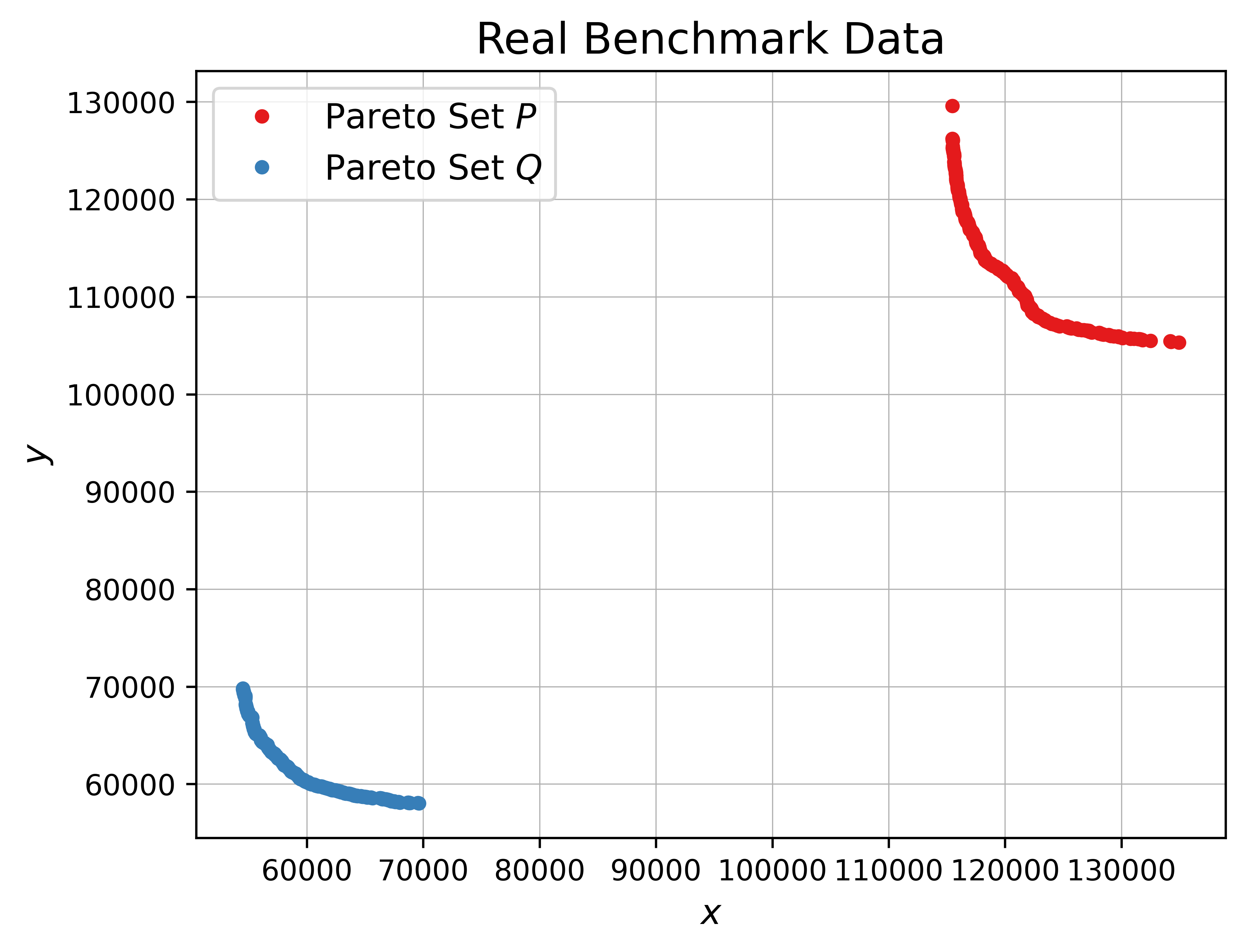} \hfill
    \includegraphics[width=0.48\linewidth]{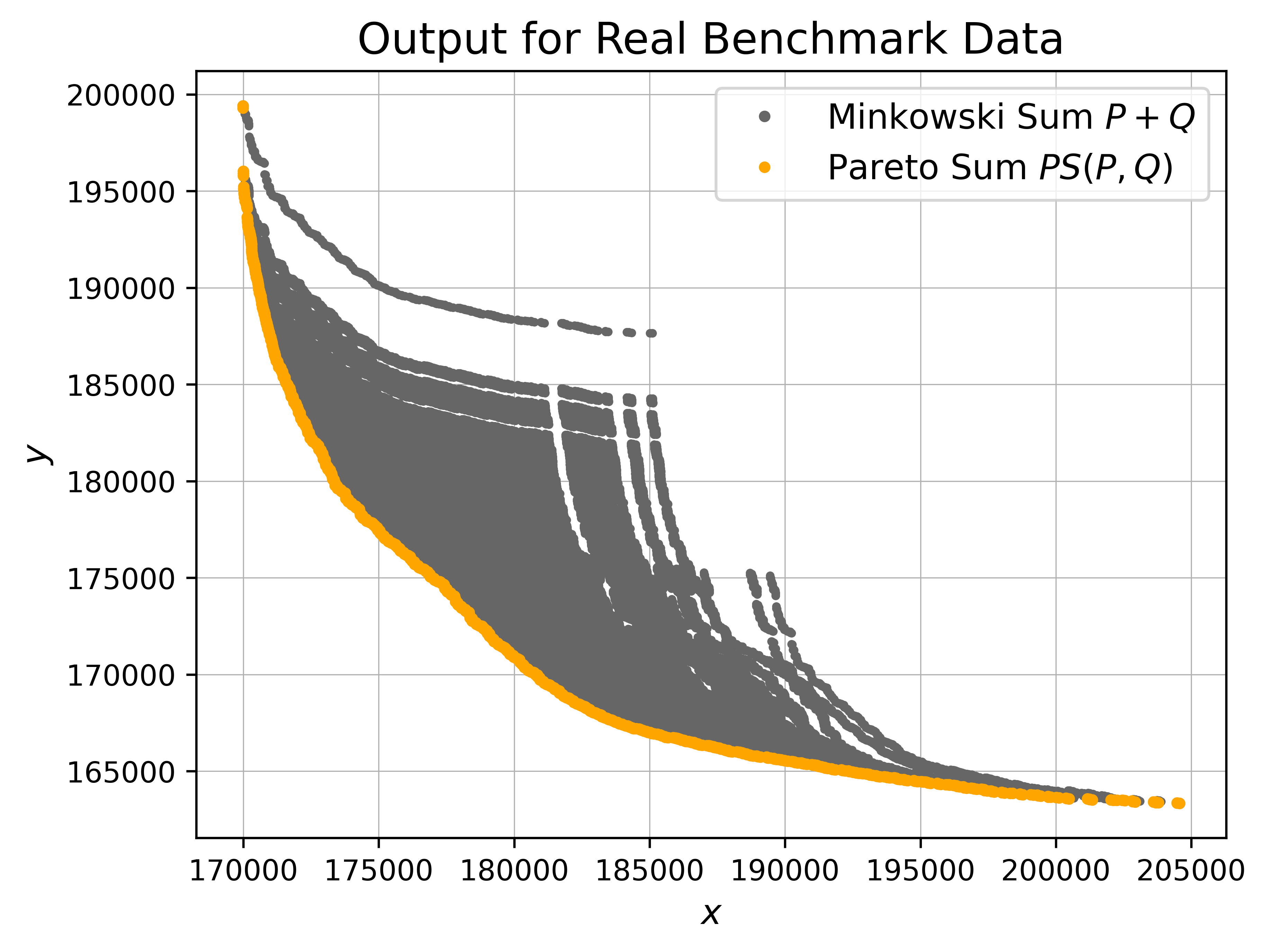}
    \caption{Two Pareto sets $P,Q$ on the left arising from bi-criteria route planning. Depicted on the right is the Minkowski sum $P+Q$ with the Pareto sum $PS(P,Q)$.}
    \label{fig:exampleIntro}
\end{figure}

Recently, exact algorithms have been developed for Pareto sum computation with a running time in $\tilde{O}(n^2)$ \cite{FunkeHSST24LB}. Naturally, the question arises whether this algorithm admits the best-possible time complexity for this problem or whether there is hope for designing asymptotically faster ones. The study of fine-grained complexity has emerged as a powerful framework for understanding the algorithmic complexity of problems within the class P. Using fine-grained reductions, the conjectured intractability of one key problem can be transferred to another problem, yielding a conditional lower bound on how fast the other problem can be solved. For the Pareto sum problem, it was recently shown that an algorithm with a running time in $\mathcal{O}(\min\{n^2,nk \}^{1-\delta}) $ for $\delta > 0$,  where $k$ denotes the output size, would disprove the 3-SUM hypothesis \cite{GokajKST25}. The conditional lower bound matches the complexity of known algorithms (up to subpolynomial factors), which renders the existence of subquadratic algorithms for Pareto sum (even with linear output sizes) unlikely. 

However, in practical applications, Pareto sum instances of substantial input  size $n$  and with output sizes $k \in \Omega(n)$   frequently arise.  On such instances, exact computation takes at least quadratic time,  which is typically impractical. This motivates our study of approximation algorithms for Pareto sum. We define a point set $\tilde{S}$ to be an approximation of $S:=PS(P,Q)$ if for each element $s$ in $S$ there is a point $\tilde{s}$ in $\tilde{S}$  of ``similar quality'' (specifically, $\tilde{s} \le s + (\Delta,\Delta)$) and each $\tilde{s}\in \tilde{S}$ is indeed close to a point in $P+Q$ (specifically $p+q\le \tilde{s} \le p+q+(\Delta,\Delta)$ for some $p\in P, q\in Q$). This provides an additive approximation guarantee. This definition caters well to many practical applications of Pareto sum computation.  For example, in bi-criteria route planning, the goal is to find the set of Pareto-optimal paths between the source and the target node with respect to two metrics. Route planning techniques for this setup heavily rely on Pareto sum computations and the number of resulting skyline points is oftentimes huge \cite{FunkeHSST24LB,goldin2021approximate, truschel2025multi}. Thus, the goal is to reduce the number of Pareto-optimal paths to be considered within the algorithm and also as an output for the user. Multiplicative approximation algorithms tend to not work well in this scenario, because for paths  with a value of zero or close to zero in one metric, even slightly suboptimal paths induce huge approximation factors. Additive approximations do not suffer from the same issue and are also much easier to interpret for a user (for example, if it is guaranteed that the reported route has at most 2 minutes of extra driving time). 
In this paper, we propose an efficient additive-approximation algorithm that allows one to trade running time against solution quality. We rigorously analyze its theoretical properties and also provide an extensive empirical study of its performance on synthetic and real data sets.

\subsection{Related Work}
The Pareto sum problem has already been extensively studied from a theoretical and practical perspective.
A simple approach follows by computing the Minkowski sum explicitly and then applying general algorithms for skyline computation~\cite{kirkpatrick1985output,chen2012maxima}. In \cite{klamroth2024efficient}, the practical relevance of the problem was emphasized, and novel algorithms were proposed that  achieve a subquadratic space consumption for certain input set structures.  Algorithms with guaranteed output-sensitive space consumption were proposed in \cite{FunkeHSST24LB}. A simple Sort \& Compare algorithm was proven to run in $\mathcal{O}(n^2\log n)$ on input sets of size $n$, and a sweep search algorithm to run in $\mathcal{O}(nk + n \log n)$ where $k$ denotes the output size. Experiments on synthetic instances and instances stemming from bi-criteria route planning in road networks confirm that these algorithms outperform those that fully compute the Minkowski sum and also the ones proposed in \cite{klamroth2024efficient}. In \cite{konen2025parameterized}, Pareto sum computation is used as a subroutine for multiobjective spanning tree or TSP computation which uses a tree decomposition as basis.  For join nodes, an algorithm similar to Sort \& Compare is used to obtain the desired combined result. In \cite{GokajK25}, several additional fine-grained complexity results for Pareto sum are provided. This includes a quadratic lower bound for Pareto sum computation even for $k = \Theta(n)$ under the 3-SUM hypothesis (which is even more believable than the previous min-plus convolution quadratic lower bound by \cite{FunkeHSST24LB}, as the min-plus convolution hypothesis implies the $3$-SUM hypothesis). Moreover, they established Pareto sum as a complete problem for a wide class of logical and geometric problems, further adding interest to the problem from the side of complexity theory.
%Moreover, they  devise an algorithm for Pareto sum computation that runs in $\tilde{O}(n^2)$ also for the higher-dimensional variant of the problem where $P, Q \subset \mathbb{R}^d$ for $d \geq 2$.
%\marvin{I wouldn't say that the higher-dimensional algorithms are particularly interesting. but we should mention the role for understanding a general class of additive problems.}

%\geri{need to move this}
%The first breakthrough for the monotone bounded min-plus convolution problem was by Chan and Lewenstein by making use of additive combinatorics techniques \cite{ChanL15}.

%The runtime has then been improved by Chi, Duan, Xie and Zhang by using FFT and hashing as a main technique \cite{ChiDX022}.
%The min-plus convolution problem has seen applications for the Tree sparsity problem \cite{MuchaW019,BackursIS17}, and the monotone bounded case has recently seen applications for the knapsack problem \cite{BringmannC22,BringmannD024}.
%Connections between bounded range and approximation algorithms have turned out useful before \cite{BringmannGK24, MuchaW019,BringmannN21}.

\subsection{Our Contributions}
In this paper, we explore the degree to which the quadratic-time bottleneck for the Pareto sum problem -- which is observable both in practice~\cite{FunkeHSST24LB} and substantiated by conditional lower bounds~\cite{FunkeHSST24LB,GokajK25, GokajKST25} -- can be circumvented by additively \emph{approximating} Pareto sums. To this end, in Section~\ref{sec:approximation}, we formally introduce a natural relaxation of Pareto sums to $\Delta$-approximate Pareto sums: intuitively, $\tilde{S}$ is an approximation of the true Pareto set $S$ of $P,Q$ if for any $s\in S$, we have an approximate point $\tilde{s}$ of ``similar quality'' (specifically, $\tilde{s} \le s + (\Delta,\Delta)$). Conversely, any such $\tilde{s}\in \tilde{S}$ is of similar quality to some actual point in $P+Q$ ("i.e.", there are $p\in P, q\in Q$ with $p+q\le \tilde{s} \le p+q+(\Delta,\Delta)$).\footnote{In fact, we often even achieve the stronger guarantee that $\tilde{S}\subseteq P+Q$, see Definition~\ref{def:approx}.}

Crucially, our results rely on establishing a certain fine-grained \emph{equivalence} between approximating Pareto sums and Bounded Monotone Min-Plus Convolution, the latter of which has received increasing interest in recent years~\cite{ChanL15,ChiDX022,BringmannC22,BringmannD024, FengJ25}. This leads to the following theoretical results:
\begin{itemize}
    \item For any sets $P,Q$ of at most $n$ points in $[0, W]^2$, we can compute an $\epsilon W$-approximate Pareto sum $\tilde{S}$ in time $\tilde{O}(n+(1/\epsilon)^{1.5})$. This leverages a remarkable algorithm due to Chi, Duan, Xie and Zhang~\cite{ChiDX022} solving Bounded Monotone Min-Plus Convolution in strongly subquadratic time $\tilde{O}(n^{1.5})$.
    \item Conversely, we show that any $\tilde{O}(n+(1/\epsilon)^{1.5-\delta})$-time algorithm for computing an $\epsilon W$-approximate Pareto sum would give a polynomial-factor improvement for Bounded Monotone Min-Plus Convolution. This establishes a fine-grained equivalence of the approximate Pareto sum and Bounded Min-Plus Convolution. Some approximation problems \cite{BringmannN21,BringmannGK24} are known to be fine-grained equivalent to min-plus convolution, but this appears to be the first such connection for bounded monotone min-plus convolution. 
    \item We observe that slightly stronger approximation guarantees can be obtained if a \emph{witness-reporting} version of Bounded Monotone Min-Plus Convolution can be achieved. Indeed, we show that such witnesses can be reported in strongly subquadratic time, by adapting a $\tilde{O}(n^{1.6})$-time algorithm for Bounded Monotone Min-Plus Convolution~\cite{ChiDX022} using techniques due to Alon and Naor~\cite{AlonN96}.
\end{itemize}

Subsequently, we observe how beneficial our relaxation to approximate Pareto sums can be on realistic input sets. To this end, we empirically fine-tune and evaluate a number of different approaches, leading to the following results:
\begin{itemize}
\item On synthetic and realistic inputs, we can approximate Pareto sums significantly more efficiently than computing them exactly. E.g., on generated data sets with $10^6$ points and coordinates in $[0,2\cdot 10^6]$, computing an additive $(\Delta=20)$-approximation saves a computation time of up to three orders of magnitude. Furthermore, the size of the returned sets reduces significantly with coarser approximation parameters. Thus, we enable an effective choice of running time/solution size vs. approximation quality.
\item We focus on efficiently solving the \emph{Bounded} Pareto sum problem, which denotes the special case in which all coordinates are integer values in $\{0,\dots, W-1\}$ -- such instances occur naturally in our approximation algorithms (after a scaling step), but are of independent interest in bi-criteria optimizations with integral cost functions (e.g., in routing). Here, we empirically evaluate and compare several algorithmic approaches: (1) a natural adaptation of the Sort \& Compare method of~\cite{FunkeHSST24LB} to BucketSort \& Compare, (2) a convex-pruning algorithm inspired by a fast algorithm for Near-convex Min-Plus Convolution~\cite{BringmannC23}, and (3) an algorithm based on a $\tilde{O}(n^{1.6})$-time algorithm for Bounded Monotone Min-Plus Convolution~\cite{ChiDX022} (which we will refer to as CDXZ algorithm). Perhaps surprisingly, we observe that depending on the class of inputs, either one of these approaches may outperform the others. This indicates that the (theoretical) connection between Pareto sums and Bounded Monotone Min-Plus convolution can also be exploited in practice to obtain practical algorithms (depending on the instances). 
\end{itemize}

We feel that our results add practical evidence (and extend theoretical evidence, e.g. \cite{MuchaW019,BringmannN21,Bringmann24}) that studying the connections of geometric (and other) problems with certain formulations of convolution-type problems is beneficial for algorithmic research.

%Changes to reduction section

\section{Bounded Pareto Sum}
In this section, we prove our core reduction which is the basis for developing new algorithms for (approximate)  Pareto sum computation. Throughout the paper, we abbreviate $[n]:=\{0,1,\dots,n-1,n\}$ and use the interval notation $[a,b]:=\{a,a+1, \dots,b-1,b\}$ for some $a,b \in \mathbb{Z}$, with $a\leq b$.
We say that a point $(x,y)$ dominates another point $(x',y')$, if $x \leq x'$ and $y' \leq y$ with at least one coordinate being strictly smaller. A point set is called a Pareto set if it contains no dominated elements. For a point $p$, we denote by $p.x$ and $p.y$ the $x$ and $y$-coordinate of $p$ respectively. The Pareto front of a set $S$, consists of the non-dominated points in $S$. We are now ready to define the problems of interest.
\begin{definition}[Bounded Pareto Sum]
Given Pareto sets  $P, Q$ of size $n$ with integer coordinates in $[W]^2$, compute the set $S$ of non-dominated points in $P + Q := \{p+q \mid ~p \in P, q \in Q\}$.
\end{definition}
We remark that for bounded Pareto sum, we always have $W \geq n-1$, as in a Pareto set the coordinates per dimension all need to be different. 
\begin{definition}[Bounded Monotone Min-Plus Convolution]
Given two monotone decreasing arrays $A$ and $B$ of size $n$ with positive entries bounded by $W=O(n)$. Compute an array $C$, where for each $k \in [2n-2]$, we have $C[k]:=\min_{i \in [n-1]} \left( A[i]+B[k-i] \right)$.
\end{definition}
The following reduction is also illustrated in Figure \ref{fig:ps_example} on a small Pareto sum instance\footnote{A similar construction has been noted by Chan and Lewenstein for the special case that $P$ and $Q$ are connected sets "i.e." all points in $P$ and $Q$ have a neighboring point with $L_1$ distance 1; see \cite[Remark~3.4]{ChanL15}}. 
\begin{restatable}{theorem}{reduce}
\label{thm:reduce}
    An instance of bounded Pareto Sum with  sets  $P, Q$ of size  $n$ and entries $(x,y) \in [W]^2$  can be reduced  in time $\mathcal{O}(W)$ to  an instance of bounded monotone min-plus convolution of two sequences $A, B$ of size $W$ with entries in $[W]$. \label{thm:reduction}
\end{restatable}
\begin{proof}
    We first perform a transformation of the $x$-coordinates of $P$ and $Q$ such that the first points in $P$and $Q$ have an $x$-coordinate of $0$. This can be achieved by simply subtracting the minimum $x$-coordinate value from each entry in the respective Pareto set. Note that this transformation does not affect the dominance relationship of any pair of points in the Minkowski sum and also the transformation can be easily remedied after having computed the Pareto sum by performing the transformations backwards.
    
    Next, we construct a min-plus convolution instance as follows. We initialize two sequences $A, B$ of size $W$ with a default value of $W$ for each entry. Then, for each $p \in P$ with $i = p.x$, we set $A[i] = p.y$. Similarly, for each $q \in Q$ with $i = q.x$, we set $B[i]= q.y$. To make $A$ and $B$ monotone, we iterate through each of them from left to right and  replace an entry at position $i$ with that of the entry at $i-1$  if the latter is smaller. Clearly, this construction  takes $\mathcal{O}(W)$ time and results in a min-plus  convolution instance with monotonically decreasing input sequences whose entries are in $[W]$.
    
    Now, let $C$ denote the min-plus convolution of $A, B$. We construct the Pareto Sum $S^*$ of $P, Q$ as a set of pairs $(k, C[k])$ for all $k \in [2W]$ where $k=0$ or  $C[k] < C[{k-1}]$. This step again requires $\mathcal{O}(W)$ time.
    To prove correctness, we first show that the set $S^*:=\{(k, C[k])|k \in [2W]\}$ is a superset of the correct Pareto Sum $S$. By definition, $C[k]$ equals $\min_{i+j = k} A[i]+B[j]$. Let now $s = p+q$ be an element of the Pareto Sum of $P, Q$ with $s.x = k$. We want to show that $C[k] = s.y$.
    By definition of the Pareto sum, $p.y+q.y$ is the smallest $y$-coordinate among all points in $P+Q$ with an $x$-coordinate of $k$. We need to prove that there are no pairs $A[i], B[j]$ with $i+j = k$ and $A[i]+B[j] < s.y$. If for all relevant $A[i]$ and $B[j]$ pairs, we have $(i, A[i]) \in P$ and $(j, B[j]) \in Q$ this follows from the definition of the Pareto Sum. Otherwise, assume w.l.o.g. that $(i, A[i]) \notin P$. Then we have  $B_i = p'.y$ where $p'$ has the maximum $x$-coordinate smaller than $i$ in $P$ (as a result of the monotonicity enforcing step).  If $A_i = p'.y$, we know that $A[{p'.x}] = A[i]$ with $p'.x < i$. As $B$ is monotonically decreasing, we thus have $B[{p'.x}] + B[{k-p'.x}] \leq A[i] + B[j]$. This argument can be iterated as long as either the current element of $A$ or that of $B$ got their final values only in the monotonicity enforcing steps. Therefore, we always end up with $A[i^*], B[{j^*}]$ where $i^* + j^* = k$, $A[{i^*}] + B[{j^*}] \leq A[i] + B[j]$ and $(i^*, A[{i^*}]) \in P$ as well as $(j^*, B[{j^*}]) \in Q$. It follows that $C[k] = p.y$ and thus also $S^* \supseteq S$.
    
    It remains to show that we prune exactly the elements from $C$ that are not in $S$.  For a pair $(k, C[k])$, we observe the following:  If $C[k] = A[i]+B[j]$ and w.l.o.g. $(i, A[i]) \notin P$ then there is a point $p \in P$ that dominates $A[i]$. Thus, $C[{p.x+j}] \leq  C[k]$ with $p.x + j < i + j$. As $C$ is monotonically decreasing, we know that $(k, C[k])$ gets pruned. The only remaining pairs are those with $C[k] = A[i]+B[j]$ where $(i, A[i]) \in P$ and $(j, B[j]) \in Q$. If for such a point we have $C[k] \geq C[{k-1}]$, we know that $(k-1, C[{k-1}])$ dominates $(k,C[k])$ and therefore $(k,C[k])$ is rightfully removed. Thus, after the pruning, only points $(k, C[k]) = p+q$ remain in the set $S$ which are non-dominated by any pair $p'+q'$ with $p' \in P,~q' \in Q$. $S^*=S$ follows. 
\end{proof}

\begin{figure}
    \centering
    \includegraphics[width=0.5\columnwidth]{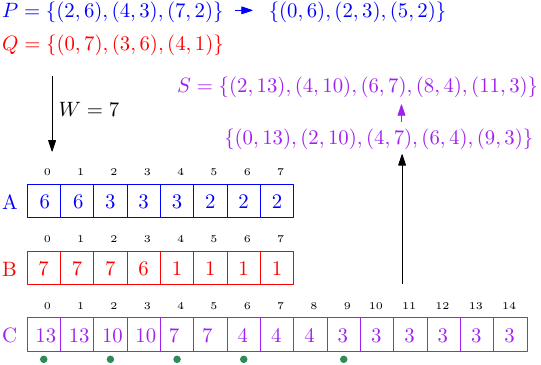}
    \caption{Reduction from Pareto sum to bounded monotone min-plus convolution and back-transformation of the solution on a small example.}
    \label{fig:ps_example}
\end{figure}
According to our theorem, if bounded monotone min-plus convolution can be solved in time $T(n,W)$ for arrays of size $n$ with entries in $[W]$, then we can solve bounded Pareto Sum  in time $O(W+T(W,W))$. This makes algorithms for bounded monotone min-plus convolution applicable to Pareto Sum with little overhead to carry out the reduction and the back-transformation of the solution. By applying the algorithm of \cite{ChiDX022}, we get:
\begin{corollary}
Given two Pareto sets $P,Q$ of size $n$ with integer entries $(x,y) \in [W]^2$, we can compute the Pareto sum of $P$ and $Q$ in time $\tilde{O}(W^{1.5}).$
\end{corollary}

\section{Additive Approximation of Pareto sum} 
\label{sec:approximation}
In this section, we study additive approximation algorithms for Pareto sum with arbitrary real-valued inputs. 

To formalize our notion of approximation, 
observe that the Pareto sum of $P$ and $Q$ can be defined as the set $S$ satisfying the following three conditions: (1) $S$ dominates $P+Q$, "i.e."., for every $p\in P$ and $q\in Q$, there exists some $s\in S$ with $s \le p+q$, (2) $S \subseteq P+Q$ and (3) $S$ is a Pareto set, "i.e.", pairwise non-dominating. To obtain an \emph{approximate} Pareto set, it is only natural to relax the first condition by allowing an additive slack of $\Delta>0$ in the domination inequality. Formally, we arrive a the following notion of approximation:

\begin{definition}[$\Delta$-Approximative Pareto sum] \label{def:approx}
Let $P,Q ,\tilde{S} \subseteq \mathbb{R}^2$. We call $\tilde{S}$ a $\Delta$-approximation of the Pareto sum of $P,Q$  if the following conditions hold:
\begin{enumerate}
 \item for all $p\in P, q\in Q$, there exists some $\tilde{s} \in \tilde{S}$ with $\tilde{s} \le p+q+(\Delta,\Delta)$.
 \item $\tilde{S} \subseteq P+Q$,
 \item $\tilde{S}$ is a Pareto set.
\end{enumerate}
\end{definition}
Let us remark that we could have arrived at this definition via slightly adapting the notion of an approximative Pareto front by the authors in~\cite{LegrielGCM10}, which leans on the notion of Papadimitriou and Yannakakis \cite{PapadimitriouY00}:
Let $d_H(S,\tilde{S})$ denote the directed Hausdorff distance between $S$ and $\tilde{S}$ under some ``distance'' $d$, "i.e.", 
$d_H(S,\tilde{S}):= \max_{s \in S} \min_{\tilde{s} \in \tilde{S}}d(s,\tilde{s}).$

Following~\cite{LegrielGCM10}, we call a set $\tilde{S}$ a \emph{$\Delta$-approximation for the Pareto set $S$} if $d_H(S,\tilde{S}) \leq \Delta$. As a natural (asymmetric) ``distance'' measure $d(s,\tilde{s})$ that quantifies the disadvantage of taking $\tilde{s}$ over $s$ (w.r.t domination of other points), we define:
$d(s,\tilde{s}) = \max \{0,\tilde{s}.x - s.x , \tilde{s}.y -s.y\}.$
Note that this measure is a more suitable definition than using, e.g., $\|\tilde{s}-s\|_\infty$, which would unnaturally punish a point that is worse by $\le \Delta$ in $x$-direction,  but \emph{much better} (by $\gg \Delta$) in $y$-direction. 
With this notion, $\tilde{S}$ is a $\Delta$-approximative Pareto sum of $P,Q$, if $\tilde{S}$ $\Delta$-approximates $S$, (2) $\tilde{S}\subseteq P+Q$, and (3) $\tilde{S}$ is a Pareto set.
\begin{figure}[h]
    \centering
    \includegraphics[width=0.4\linewidth]{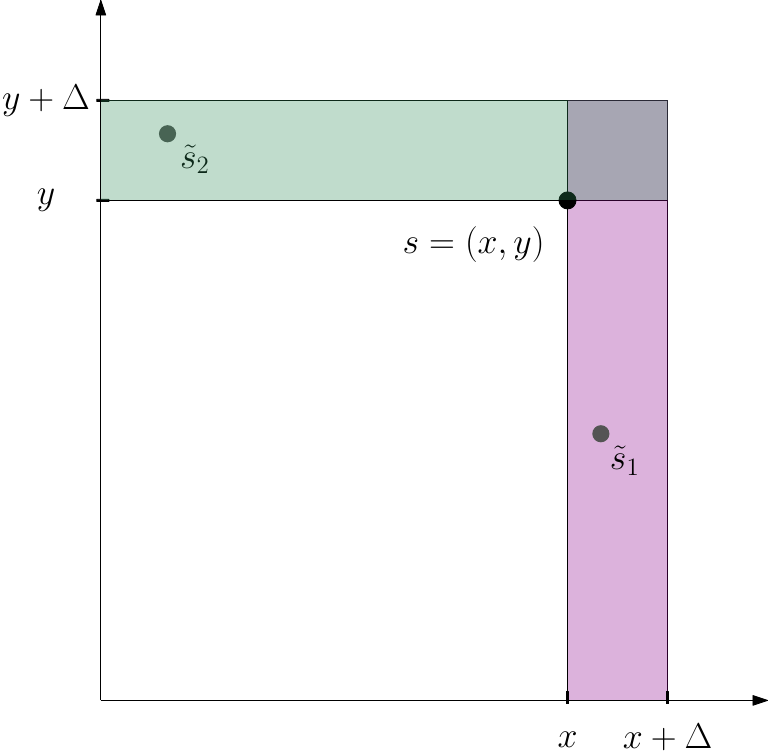}
    \caption{The points $\tilde{s}_1$ and  $\tilde{s}_2$ are examples of  points with $d(\tilde{s}, s) \leq \Delta$ for the shown $s \in S$. Note that if $s\in S$, then no point in $P+Q$ dominates $s$ ("i.e.", no point lies to the lower left of $s$).}
    \label{fig:approx}
\end{figure}

To obtain such an approximation, we proceed as follows: we first scale down the points of $P$ and $Q$ by dividing the $x$ and $y$-coordinates by a scaling parameter~$t$, creating the sets $P'$ and $Q'$. We additionally save a \emph{representative} $\rep(p')\in P,\rep(q')\in Q$  for each (rounded) point $p'\in P'$ and $q'\in 
Q'$ in the original sets (we choose the one with minimal $L_1$ norm but any representative works). Then we compute the Pareto sum $S'$ for the sets $P'$ and $Q'$ and determine, from each $s'=p'+q'\in S'$, an approximate point $\tilde{s}$ using the representatives of $p'$ and $q'$, which yields an approximate set $\tilde{S}$. As we shall prove, this simple idea turns out to be conditionally optimal, by devising a fine-grained equivalence with bounded $(\min,+)$-convolution. This algorithm is formalized in Algorithm \ref{alg:add_approx}, which takes as input $P, Q$ and a scaling factor~$t$, and returns a $2t$-approximation of the Pareto sum of $P$ and $Q$.
\begin{algorithm}
\caption{ApproximateParetoSum($P,Q,t$)}
\label{alg:add_approx}
\begin{algorithmic}
\STATE{$P' \gets \left \{ \left(\lfloor \frac{p.x}{t} \rfloor , \lfloor \frac{p.y}{t} \rfloor \right)\mid p\in P \right\}$}
\STATE{$Q' \gets \left \{ \left(\lfloor \frac{q.x}{t} \rfloor , \lfloor \frac{q.y}{t} \rfloor \right)\mid q\in Q \right\}$}
\FOR{$p' \in P'$}
\STATE{$\rep(p') \gets  \arg \min_{p \in P} \left \{\lVert p \rVert_{1} \mid p'= \left( \left  \lfloor   \frac{p.x}{t}  \right \rfloor, \left \lfloor   \frac{p.y}{t}  \right \rfloor \right)  \right \}$}
\ENDFOR
\FOR{$q' \in Q'$}
\STATE{$\rep(q') \gets  \arg \min_{q \in Q} \left \{\lVert q \rVert_{1} \mid q'=\left( \left  \lfloor   \frac{q.x}{t}  \right \rfloor, \left \lfloor   \frac{q.y}{t}  \right \rfloor \right)   \right \} $}
\ENDFOR
\STATE{$P'\gets $ Pareto front of $P'$, $Q' \gets$ Pareto front of $Q'$ }
\COMMENT{\emph{if necessary}}
\STATE{$S' \gets PS(P',Q')$, where $s'\in S'$ is given as $(p',q')$ with $s'=p'+q'$ and $p'\in P', q'\in Q'$}
\STATE{$Z \gets\{ \rep(p')+\rep(q') \mid (p',q') \in S'\}$}
\RETURN $\tilde{S} \gets $ Pareto front of $Z$
\end{algorithmic}
\end{algorithm}
\begin{lemma}
    For any $P,Q$ and $t>0$, let $S$ denote $PS(P,Q)$ and let $\tilde{S}$ denote the result of ApproximateParetoSum($P,Q,t$). Then $\tilde{S}$ is a $2t$-approximation of the Pareto sum of $P,Q$. %For any $c\in C$, there is $\tilde{c}\in \tilde{C}_t$ such that $c\leq \tilde{c} \le c+2\mathbf{t}$, where $\mathbf{t} = (t,t)$. 
    \label{lem:bounding_approx}
\end{lemma}
\begin{proof}
    We first argue that the set $Z$ computed in Algorithm~\ref{alg:add_approx} satisfies the relaxed domination condition with $\Delta = 2t$, specifically, that for all $p+q$ with $p\in P,q\in Q$, there exists some $z=\rep(p')+\rep(q')\in Z$ with $z \le p+q+(2t,2t)$. 
    Indeed, for any $p\in P, q\in Q$, consider their rounded points $p'= \left(\lfloor \frac{p.x}{t} \rfloor , \lfloor \frac{p.y}{t} \rfloor \right) \in P'$, $q'= \left(\lfloor \frac{q.x}{t} \rfloor , \lfloor \frac{q.y}{t} \rfloor \right)  \in Q'$. By definition of the Pareto sum $S'$ of $P',Q'$, there must exist some $(\tilde{p}+\tilde{q})\in S'$ with $\tilde{p}+\tilde{q} \le p'+q'$. We claim that $z := \rep(\tilde{p})+\rep(\tilde{q})$ approximately dominates $p+q$, since
    \begin{align*}
    z &= \rep(\tilde{p}) + \rep(\tilde{q})  \le (t \tilde{p} + (t,t)) + (t\tilde{q} + (t,t))\\
     &= t(\tilde{p}+\tilde{q}) + 2(t,t) \le t(p'+q')+2(t,t) \\
    & \le p+q+2(t,t),
    \end{align*}
    where we used $\rep(x') \le tx' + (t,t)$ for any rounded point $x'$ in the first line, and $tx' \le x$ for any $x$ and rounded point $x'$ in the third line.
    
    We claim that the lemma follows: Filtering out dominated points in $Z$ ensures that $\tilde{S}$ is a Pareto set (establishing Condition 3.), without affecting the domination condition (if $z\in Z$ dominates $z'\in Z$, then whenever $z'\le p+q+(t,t)$, then $z\le z' \le p+q+(t,t)$, so $z'$ is never needed to ensure the condition). Finally, clearly $\tilde{S}\subseteq Z \subseteq P+Q$ as any element of $Z$ is explicitly given as $\rep(p')+\rep(q')$ for corresponding $\rep(p')\in P$ and $\rep(q')\in Q$. 
    \end{proof}
\begin{theorem}\label{thm:delta_approx}
Let $T(n,W)$ denote the runtime of monotone min-plus convolution of arrays of size $n$ and entries bounded by $W$ with witness reporting. Given sets $P,Q$ of $n$ points in $\mathbb{R}^2$, we can compute a $\Delta$-approximate Pareto sum in time $\mathcal{O}(n + (W/\Delta) \log(W/\Delta) + T(W/\Delta,W/\Delta) )$.  
\end{theorem}
\begin{proof}
Given Lemma~\ref{lem:bounding_approx}, it remains to show that ApproximateParetoSum($P,Q,t)$ with $t:= \Delta/2$ can be computed in the desired running time. 
After $O(n)$-time preprocessing, the main step in Algorithm~\ref{alg:add_approx} is to compute the Pareto sum $S'$ of $P',Q'$. By our rounding, $P',Q'$ have integer coordinates in $\{0,\dots, \lfloor W/\Delta \rfloor\}$, and Theorem~\ref{thm:reduction} is applicable to obtain $S'$ and thus $Z$ in time $T(\lfloor W/\Delta \rfloor , \lfloor W/\Delta \rfloor )$.
Observe that $Z$ is a set of size at most $O(W/\Delta)$. Thus, we can compute its Pareto front in time $O((W/\Delta) \log (W/\Delta))$, concluding the claim.
\end{proof}
%it remains to compute the Pareto front of $Z$

%The runtime follows by computing the skyline in the last step by sort and compare, and reducing the Pareto sum computation to the bounded monotone min-plus convolution problem. The sets $A',B'$ will be of size at most $2n/\Delta$ as the original sets $A,B$ only contain at most one point for each $x$-coordinates.

%For the correctness, we need to check the conditions given by Definition \ref{def:approx}. 

%For Condition 1, we bound $d_H(PS(A,B), \Tilde{C}_\Delta)$. By Lemma \ref{lem:bounding_approx}, we can choose for every $c\in PS(A,B)$, a point $\Tilde{c}$, with the property that $ c \leq \Tilde{c} \leq c+2\mathbf{\Delta}$, where $\mathbf{\Delta}:=(\Delta/2,\Delta/2)$.
%Then 
%$\left|c-\Tilde{c} \right|_{\infty}  \leq |2\mathbf{\Delta}|_{\infty}=\Delta.$

%Condition 2 holds as for each $a'+b' \in C'$ we have $\rep(a') \in A, \rep(b') \in B$ by definition.

%For Condition 3, we argue why we can safely remove dominated points, and thereby construct a Pareto set. Assume there are points $p,q \in \tilde{C}_t$, with $p \leq q$. Let $c$ be any point, which satisfies $ q \leq c +2t$ then by transitivity also $p \leq c+2t$, so that we can safely remove point $q$ from the set $\tilde{C}_t$.

\begin{corollary}
Given sets $P,Q$ of $n$ points with real coordinates in $[0,W]^2$ and $\epsilon > 0$, we can compute an $\epsilon W$-approximation of the Pareto sum in time $\ \tilde{\mathcal{O}} \left( n+ (1/\epsilon)^{1.6} \right)$. \label{thm: epsW_runtime}
\end{corollary}
\begin{proof}
Follows from the $\tilde{O}(n^{1.6})$ time algorithm of \cite{ChiDX022}, and our witness reporting extension, described in Section \ref{sec: witness_reporting} in the appendix\footnote{We remark that it is very plausible that our witness reporting algorithm extends to the $\tilde{O}(n^{1.5})$ algorithm.}. 
\end{proof}
Let us also now show a lower bound, which will nicely complement Algorithm \ref{alg:add_approx}. 
\begin{theorem}\label{thm: lower_bound}
Let $T_{apx}(n,W, \Delta)$ denote the running time of computing, given Pareto sets $P,Q$ in $\mathbb{R}^2_{\ge 0}$ with maximum coordinate bounded by $W$, a $\Delta$-approximation of their Pareto sum, with witnesses. If there exists $\alpha \ge 0$ and $c$ such that $T_{apx}(n,W,\Delta) \le O(n+(W/\Delta)^{c})$ whenever $\Delta = \Theta(W^\alpha)$, then bounded monotone min-plus convolution can be computed in time $O(n^{c})$, with witnesses. 
\end{theorem}
\begin{proof}
The proof is an adaptation of the reduction from min-plus convolution  to Pareto sum given in~\cite{FunkeHSST24LB}. %, which ensures that a $\Delta$-approximation still allows to decide the given instance correctly.
Consider an arbitrary min-plus convolution instance $A[0,...,n-1],B[0,...,n-1]$ with monotonicity (specifically, $A[i] > A[i+1], B[i] > B[i+1]$ for all $i$) and entries in $\{0,\dots, O(n)\}$. We construct the sets of points 
\begin{align*}
    &P := \{(i\cdot (2n^{\frac{\alpha}{1-\alpha}}),A[i]\cdot (2n^{\frac{\alpha}{1-\alpha}})) \mid i \in [n-1]\}, \\
    &Q  := \{(i\cdot (2n^{\frac{\alpha}{1-\alpha}}), B[i] \cdot (2n^{\frac{\alpha}{1-\alpha}}) )\mid i \in [n-1] \}.
\end{align*}
Consider a $\Delta$-approximation $\tilde{S}$ of the Pareto sum of $P,Q$ with $\Delta =  n^{\frac{\alpha}{1-\alpha}} $ . By construction of the point sets, if a point $\tilde{s}\in \tilde{S}\subseteq P+Q$ $\Delta$-dominates a point $p+q$, "i.e.", $\tilde{s} \le p+q+(\Delta,\Delta)$, then it already dominates $p+q$, "i.e.", $\tilde{s} \le p+q$. Thus, $\tilde{S}$ is the true Pareto sum of $P,Q$, and we can extract the min-plus convolution $C$ of $A,B$: By construction, for every $k\in [2n-2]$, there must be a point in $PS(P,Q)$ with $x$-coordinate $k\cdot (2n^{\frac{\alpha}{1-\alpha}})$ and $y$-coordinate $(\min_{i \in [n-1]}A[i]+B[k-i])\cdot (2n^{\frac{\alpha}{1-\alpha}})$, revealing $C[k]=\min_{i \in [n-1]}(A[i]+B[k-i])$. The existence of a point with $x$-coordinate $k\cdot (2n^{\frac{\alpha}{1-\alpha}})$ for every $k \in [2n-2]$ follows by the decreasing values of $C[k]$. Also, witnesses for such points can be transferred.

Note that the constructed instance of Pareto sum Approximation consists of $n$ points in $P,Q$ with coordinates bounded by $W\le O(n^{1+\frac{\alpha}{1-\alpha}}) =O(n^{\frac{1}{1-\alpha}}) $ and $\Delta = n^{\frac{\alpha}{1-\alpha}}$. Thus, if $T_{apx}(n,W,\Delta) \le O( n+(W/\Delta)^{c})=O(n+(n^{\frac{1}{1-\alpha}-\frac{\alpha}{1-\alpha}})^{c})$, then we can compute the min-plus convolution of $A,B$ in time $O(n^{c})$, with witnesses. 
\end{proof}

\begin{remark}
Even without witness reporting, we still get a fine-grained equivalence (for slightly weaker guarantees), by a minor adaptation of Algorithm \ref{alg:add_approx} and Theorem~\ref{thm: lower_bound}. We relax the condition $\tilde{S} \subseteq P + Q$ to the milder condition that  each $\tilde{s}$ is in the vicinity of a point $p+q \in P+Q$. More formally, for all $ \tilde{s} \in \tilde{S}$ there exist $p \in P$ and $q \in Q$ such that $p+q \leq \tilde{s} \leq p+q+2(t,t).$
The details can be found in Section \ref{sec:without_witnesses}.
\end{remark}

\section{Framework and Implementation} \label{sec:implement_det} 
Our proposed additive approximation algorithm, presented in  Algorithm \ref{alg:add_approx}, consists of scaling down the instance, computing the Pareto sum of the transformed point sets, and then scaling up the solution. The Pareto sum instance after scaling is captured by our definition of bounded Pareto sum. To compute its solution, we can obviously directly apply any Pareto sum algorithm. By virtue of Theorem \ref{thm:reduction}, however, also any algorithm for bounded min-plus convolution translates to an algorithm for bounded Pareto sum when performing the required reduction and back-transformation of the solution. Thus, depending on whether we use an algorithm directly designed for Pareto sum or a convolution algorithm, the computation pipeline needs to be adapted. Figure \ref{fig:pipeline} provides an overview of our framework.
\begin{figure}
    \centering
    \includegraphics[width=1.0\linewidth]{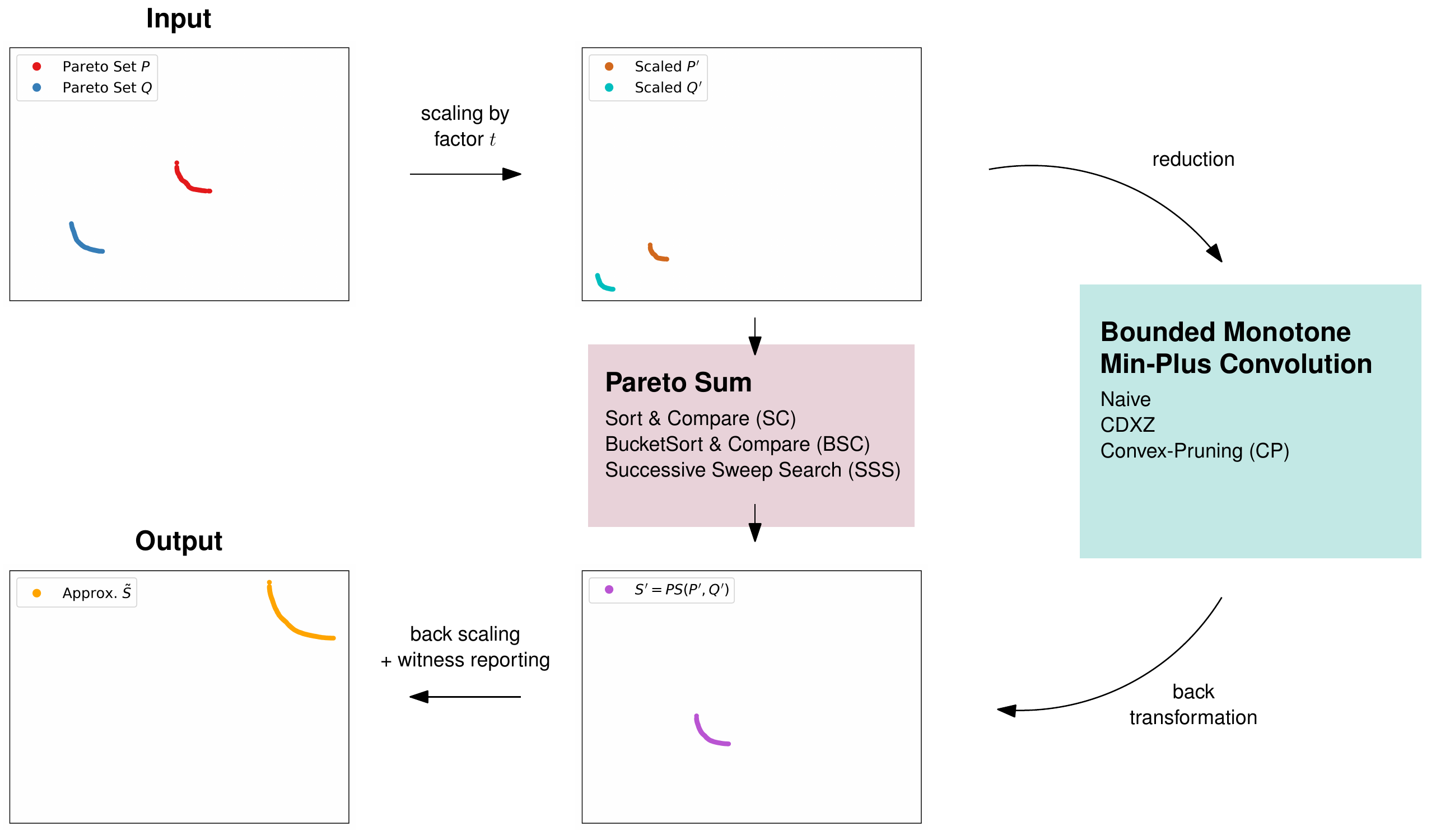}
    \caption{Overview of the main steps of our additive approximation algorithms. An important step is to compute the Pareto sum of the scaled instance. This can either be done directly via a Pareto sum algorithm or by an algorithm for bounded monotone min-plus convolution when applying the proper forward and backward reduction steps.}
    \label{fig:pipeline}
\end{figure}
In the following, we describe the  algorithms of each category in more detail.
\subsection{Algorithms for (bounded) Pareto sum}
We start with the algorithms designed directly for Pareto sum computation. 
\paragraph*{Sort \& Compare (SC) and Successive sweep search (SSS)}
We implemented both state of the art algorithms for (general) Pareto sum computation from  \cite{FunkeHSST24LB} as a baseline. In a nutshell, the Sort \& Compare algorithm sorts the points in the sumset $P+Q$ lexicographically and performs constant time dominance checks per point based on this ordering. Further details on the baselines can be found in Section \ref{sec:direct_algorithms} in the appendix.
\paragraph*{BucketSort \& Compare (BSC)}
For bounded integral Pareto sum, we can improve the Sort \& Compare algorithm, by  using bucketsort to sort the points,  directly updating the minimal $y$-coordinate for a corresponding $x$-coordinate. Further details can be found in section \ref{sec:bucketsort_and_compare} in the appendix. 

\subsection{Algorithms for min-plus convolution}
After the reduction to min-plus convolution, we can of course use the naive quadratic time algorithm for the latter. In the following, we discuss two more sophisticated and (potentially) faster convolution algorithms. 
\paragraph*{CDXZ Algorithm }
We discuss a simplified version of the $\tilde{O}(n^{1.6})$-algorithm \cite[Section 4.1]{ChiDX022} for monotone bounded min-plus convolution in a detailed fashion in Section \ref{sec:Chi_Algorithm} in the appendix. 
Using this algorithm in our Pareto sum reduction framework results in a running time of $\mathcal{O}(W^{1.6} \log W)$ for bounded Pareto sum.
Let us highlight the main improvements and simplifications we achieved, which are also presented in Section \ref{sec:Chi_Algorithm} in the Appendix.
\begin{itemize}
\item We removed the segment tree data structure. Making use of our simpler enhanced min-plus convolution algorithm instead of using a segment tree with interval updates is simpler and improves the runtime theoretically and practically.
\item We improved the bound of $b$, which was previously ranging from $b \in \{0,1,2\}$ to $b \in \{0,1\}.$ The better bound of the constant $b$, is insignificant in the theoretical analysis of the algorithm, but it plays an important role in the error correction part. Besides the initial polynomial multiplication of $P$ and $Q$, every other part of the algorithm in the error correction phase will be repeated $b$-times, thus the better bound on $b$ will result on repeating these costly parts of the algorithm twice instead of thrice.
\item As the set $T_b$ can be of expected size $\mathcal{O}(n^{1.6})$ (for the choice of $\beta=0.4)$, we do not save it explicitly like \cite{ChiDX022} suggests. In order to circumvent storing the set $T_b$, we make an online computation of the polynomials $R_{p,b}[k](x)$ simultaneously, updating the polynomials after having seen a new pseudo-witness $(i,k-i)$ on the fly.
\end{itemize}

We implemented the version without witness reporting, as the witness reporting algorithm would be too inefficient in practice.
This leads to a weak approximation of the Pareto sum; details can be found in Section \ref{sec:without_witnesses} in the appendix.
For the multivariate polynomial multiplication, we perform a standard reduction to the univariate case, by turning $P$ and $Q$ into univariate polynomials. 
Instead of FFT, we make use of the polynomial multiplication offered by the library FLINT \cite{flint}, where the univariate polynomials are kept as lists.
Instead of using binary search trees in the computation of $T_b$, we performed a linear search, after saving for each possible modulo value $ i \in [p-1]$ the indices $k$ with $C'[k] +b\mod p =i$. This turned out to be faster in the experiments.
The parameters $\alpha$ and $\beta$ are not fixed as the optimal theoretical runtime suggests. The choice of the parameters $\alpha$ and $\beta$ are transferred to the choice of the scaling factor $n^\alpha$ and the prime $p$, which are now inputs of the implementation; see Figure \ref{fig:parameterChi} for the influence of these parameters on the runtime.
\paragraph*{Convex Pruning (CP) algorithm}
Bringmann and Cassis developed an algorithm for min-plus convolution, suited towards the case, where the arrays $A$ and $B$ are near-convex \cite{BringmannC23}. We implemented two versions of their algorithm with and without witness reporting. Our implementation is close to the descriptions in \cite{BringmannC23}, with the main difference being that we did not use sparse convolutions, but rather an enhanced min-plus convolution computation, which works particularly well when the arrays $A$ and $B$ have few distinct entries. We precomputed groups of entries of the same value in $A$ and $B$. More details on the enhanced min-plus convolution can be found in Algorithm \ref{alg:enhanced_min_plus} and further details on the precise changes can be found in Section \ref{sec:convex_pruning} in the appendix.

\subsection{Further engineering}
Algorithm \ref{alg:add_approx} has been implemented mostly matching its description, with small changes, which we describe now. 
When using a  convolution algorithm within Algorithm \ref{alg:add_approx}, we also perform the needed reduction from bounded Pareto sum to montone bounded min-plus convolution and the transformation backwards.
The reduction closely follows the constructive proof of Theorem \ref{thm:reduction}, the backwards transformation can be achieved very simply with a standard 2-pointer approach.
At the  end of Algorithm \ref{alg:add_approx}, instead of computing the Pareto front of the set $Z$ by sorting, we do the following: For each newly inserted point, we look back and remove all points, which are dominated by this new point. 
This will work fast as the set $Z$ is almost a Pareto set, so few corrections need to be made. 

\section{Experimental Evaluation} 
In the following, we evaluate our approximate algorithm with the different choices for bounded Pareto sum computation sketched above. While all algorithms are described for the case that $P$ and $Q$ are of the same size $n$, they also all work for differing sizes.
Our  code is written in C++. It was compiled with gcc using -O3. In our implementation of the CDXZ algorithm, we use routines from  FLINT  (Fast Library for Number Theory) \cite{flint}.
All experiments were conducted on a single core of an AMD Threadripper 3970 with 256 GB of RAM. We set a time-out of one hour per run. Running times are always averaged over 10 runs. 
\bigskip\\
\textbf{Benchmarks.}
We use three types of benchmarks in our evaluation. Two are synthetic to allow us to control instance properties, while the third benchmark consists of real data.
\bigskip\\
\emph{Benchmark 1: Instances with bounded range.}
First, we use an instance generator that allows us to fix $n$ and $W$. Here, it is important to ensure $W \geq  n-1$ to guarantee that there are enough distinct natural numbers in the range to construct a Pareto set. We use $W=c \cdot n$ for $c \geq 1$. To efficiently create a strictly monotonically increasing sequence, we initialize an array with a counter $\gamma_i$ for each $i \in [W]$. Then, we increment a randomly selected counter in each of $n$ rounds. Afterwards, we sweep over this array, starting at position 0, and perform the following operations: If $\gamma_i > 1$, we set $\gamma_i=1$ and  $\gamma_{i+1} = \gamma_{i+1}+\gamma_i-1$. Here, entry 0 is assumed to be the successor of $\gamma_W$. We stop when all entries are 0 or 1 which happens after each entry has been considered at most twice. Then, for all $\gamma_i=1$, element $i$ is added to the sequence. This produces a valid sequence in time $\mathcal{O}(n+W)$. Pareto sets can then be constructed by using two such sequences, one sorted increasingly for the $x$-coordinates and the other one decreasingly for the corresponding $y$-coordinates.
\bigskip\\
\emph{Benchmark 2: Function-based instance generator.}
For any strictly monotonically decreasing function $f$ defined on $\mathbb{R}^+$, sampling $n$ points $(x, f(x))$ provides us with a Pareto set. We use two base functions: the first uses $f(x)=-x$ and the second $f(x)=\frac{c}{x}$ with a constant $c \in \mathbb{R}^+$. We  slightly perturb the points to have more input variability, while guaranteeing that pairwise non-dominance is upheld. We refer to the resulting instances as near-linear or near-curved. \begin{figure}
    \centering
    \includegraphics[width=0.45\linewidth]{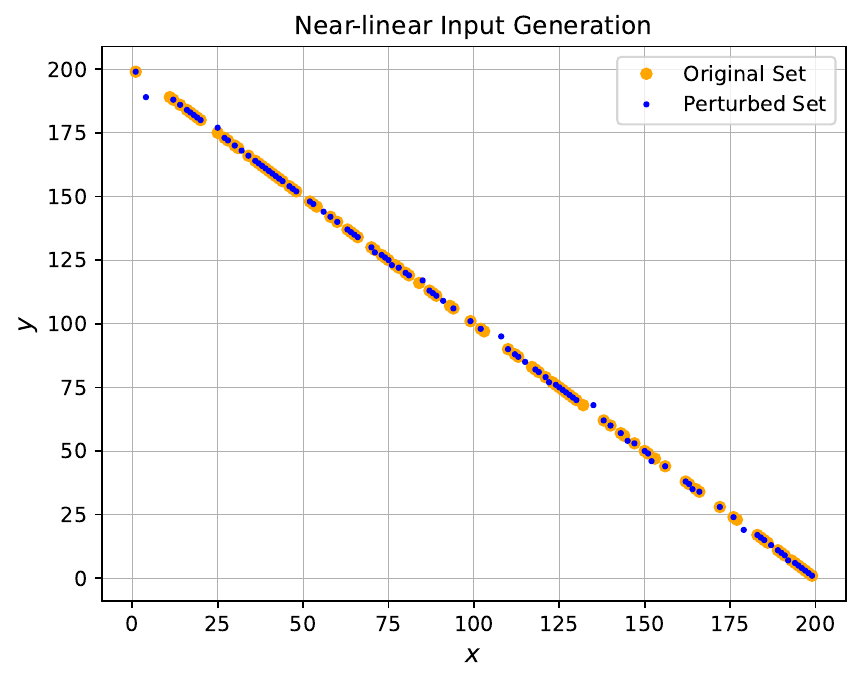} \hfill
    \includegraphics[width=0.45\linewidth]{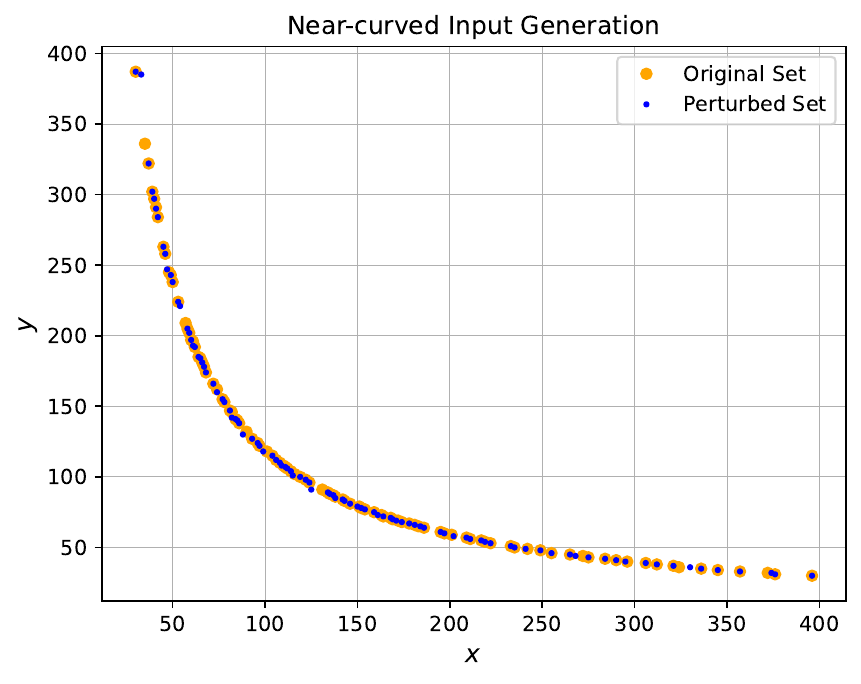}
    \caption{Two different versions of generating input sets with values in the range $[0,2n]$: Near-linear and near-curved input sets are based on the linear and curve generator, respectively, where we apply perturbation to a subset of the points.}
    \label{fig:inputVariants}
\end{figure}See Figure \ref{fig:inputVariants} for an illustration.
We note that by means of the  perturbation, the resulting point sets  are no longer convex.
\bigskip\\
\emph{Benchmark 3: Route planning instances.}
Finally, we also use real-world data sets for evaluation. The instances stem from bi-criteria route planning in line with \cite{FunkeHSST24LB}. Here, given a graph in which the edges are augmented with  two cost values (in our case distance in meters and positive height difference in centimeters), the goal is to compute all Pareto-optimal paths from a source node $s$ to a target node $t$.  Using bi-directional search to answer such queries, or constructing a data structure which accelerates query answering \cite{storandt2012route}, demands to combine and filter the set of Pareto-optimal paths from $s$ to some intermediate node $v$, and from $v$ to the target $t$. The cost pairs of these two path sets form our inputs $P$ and $Q$. The benchmark contains roughly 100000 instances based on randomly chosen $s,v$ and $t$ triples in a directed road network with one million nodes and two million edges. Note that the sizes of $P$ and $Q$ might differ significantly in each instance.
\bigskip\\
\textbf{Results for Exact Pareto Sum.} \label{sec:comparative_results}
For a general comparison of the different algorithms  for Bounded Pareto sum discussed in this paper, 
we first applied a scalability study on the range-bound instances with all point coordinates being integers, see Figure \ref{fig:timeAll}.
\begin{figure}
    \centering
    \includegraphics[width=0.8\linewidth]{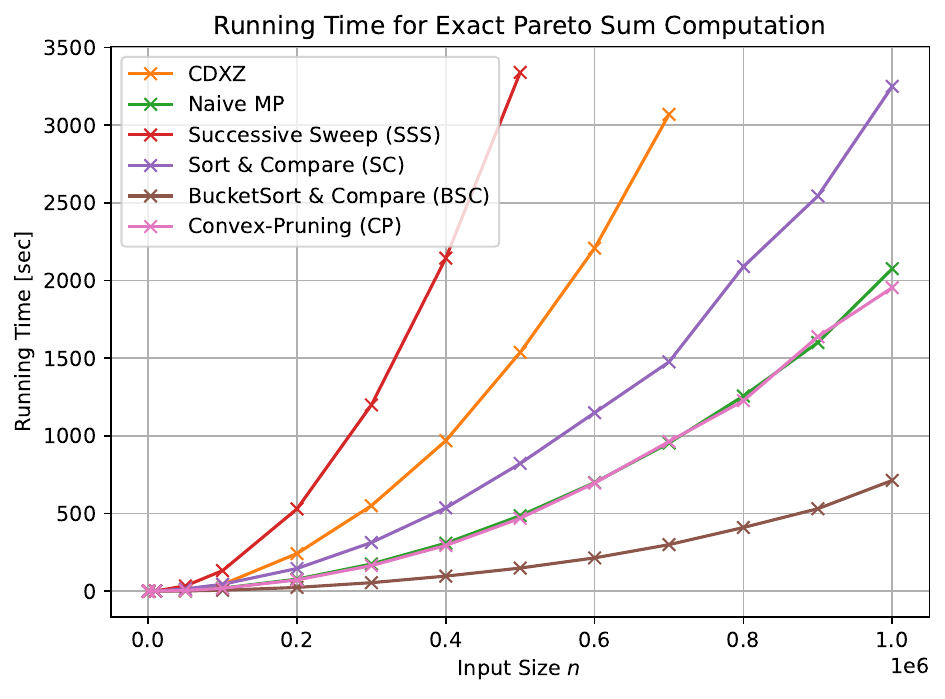}
    \caption{Running times for Pareto sum computation on generated inputs of size up to $n= 100,000$ with values in the range of $[0,2n]$ and a timeout of 1 hour.}
    \label{fig:timeAll}
\end{figure}
We observe huge performance differences. SSS is by far the slowest approach. Its theoretical performance is only subquadratic in case the output size is in $o(n)$. However, for the considered instances there are roughly $4n$ output points. The BSC algorithm performs best and also significantly outperforms the plain SC algorithm. Even more remarkably, also the min-plus convolution based algorithms, naive and CP, outperform SC. While CDXZ beats SSS, it is a bit slower than SC on this instance type. As described in Section \ref{sec:implement_det}, our implementation of the CDXZ algorithm  allows us to choose its  parameters. We conducted a series of tests on synthetic instances of different sizes  to determine a good parameter choice, see Figure \ref{fig:parameterChi} for an example. 
\begin{figure}
    \centering
    \includegraphics[width=0.8\linewidth]{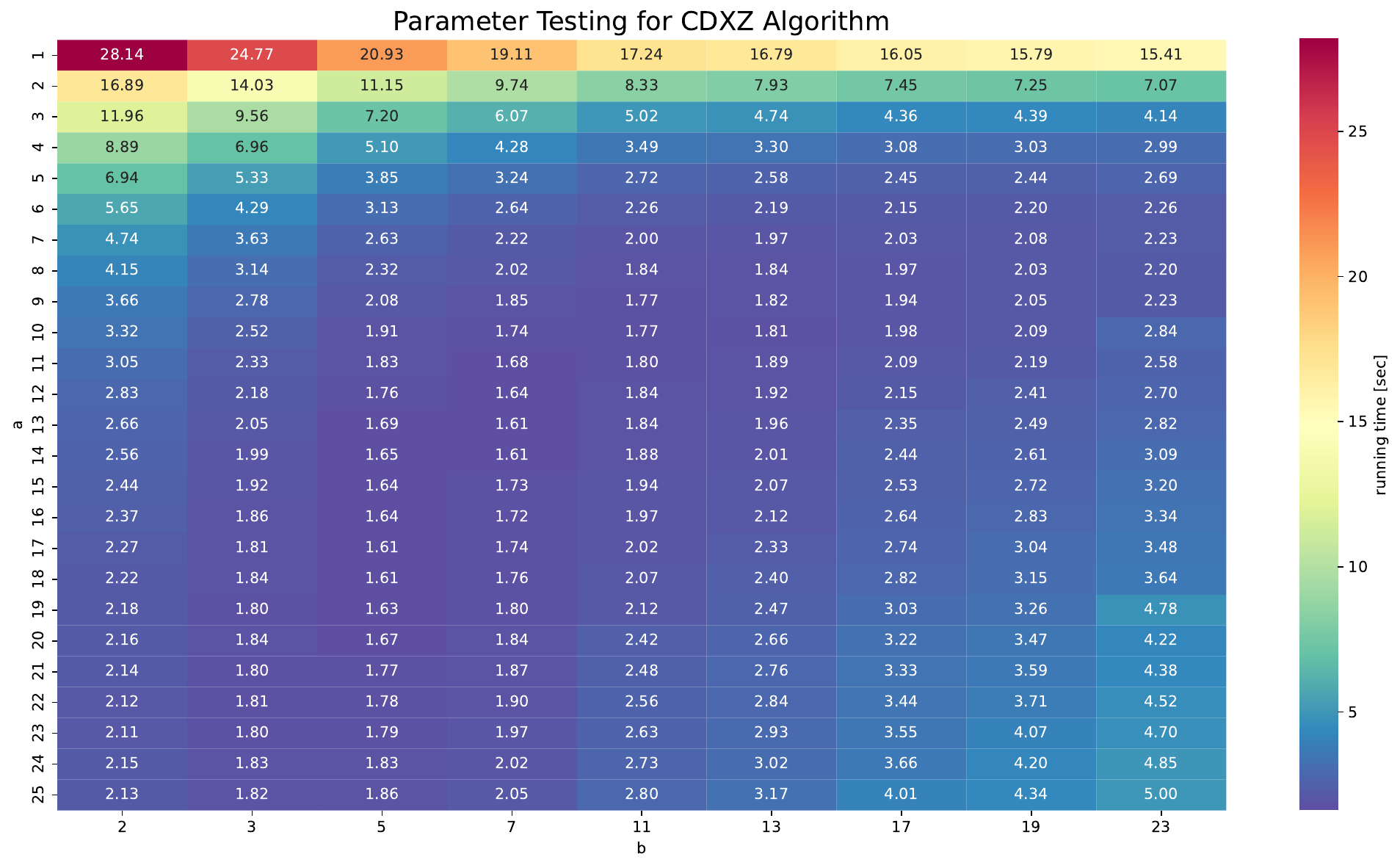}
    \caption{Parameter testing for the CDXZ algorithm on inputs of size $n=20000$ in the range $2n$.}
    \label{fig:parameterChi}
\end{figure}
We observe that the parameter values heavily impact the performance. Across all tested instances, scaling factors of $25$ and $p=2$ were among the top configurations. We therefore use those values in all subsequent experiments. We remark, though, that additional fine-tuning or choices adaptive to $n$ could potentially further improve the performance.

However, our main interest is in computing approximate   solutions. In the following, we  investigate how the algorithms perform when used as a plug-in in our approximation framework. In the remainder of the experiments, we focus on the BSC algorithm as well as the convolution based algorithms, as they outperform the baseline algorithms.
\bigskip\\
\textbf{Results for Approximate Pareto Sum.}
Next, we present comparative running time results of our approximation algorithm with different inner methods of Pareto sum computation applied to the scaled down inputs.
For fair comparison, we use the non-witness reporting variants for all of them.
Figures \ref{fig:approxRangeGen}, \ref{fig:approxNearLinear}, and \ref{fig:approxNearCurve} depict the results for range-bounded, near-linear and near-curve instances, respectively, using a value range of $2n$ and a scaling factor of $t=10$. Interestingly, for each input type, another algorithm performs best. For range-bounded instances, BSC is the fastest with CDXZ being second, while on near-linear instances CDXZ outperforms the other two methods. On near-curved instances, however, CP is clearly the best, beating the others by at least an order of magnitude. 
\begin{figure}
    \centering
    \includegraphics[width=0.7\linewidth]{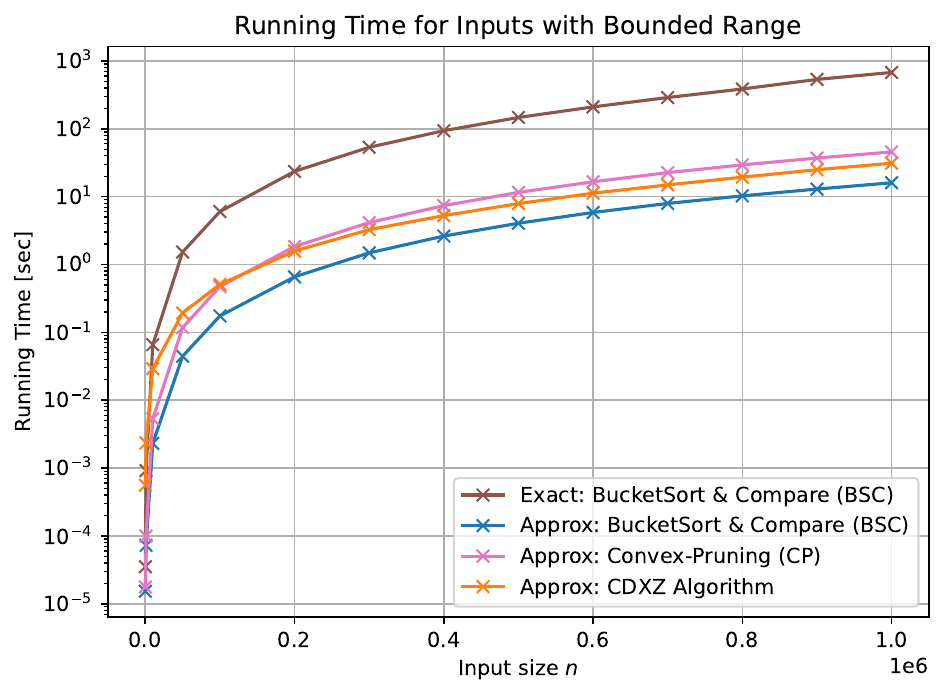}
    \caption{ Running time of the approximation variants on inputs generated with bounded range of $2n$ using a scaling factor of $t=10$. }
    \label{fig:approxRangeGen}
\end{figure}
\begin{figure}
    \centering
    \includegraphics[width=0.6\linewidth]{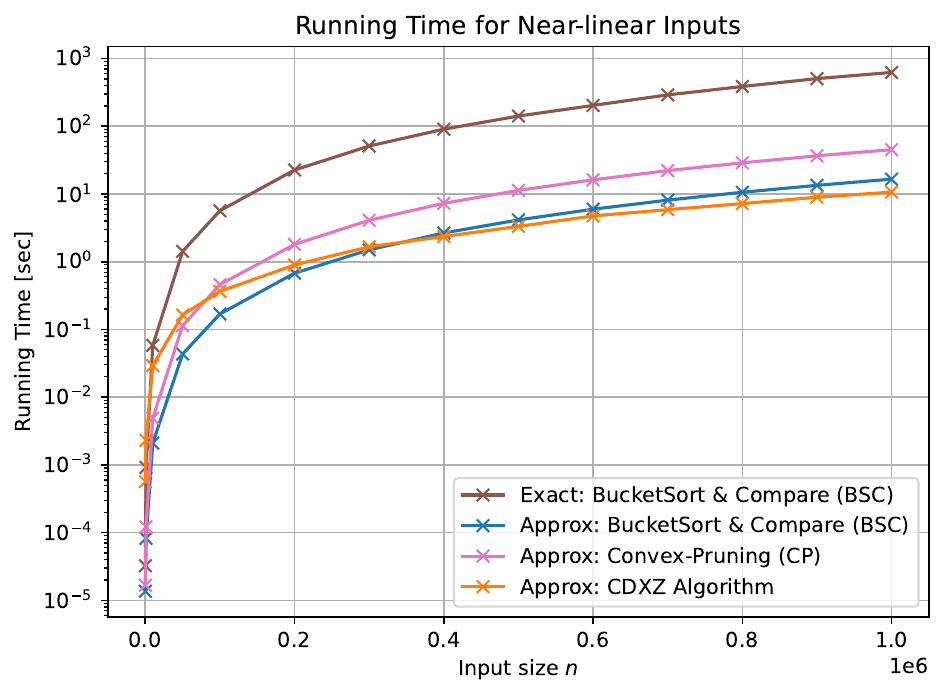}
    \caption{Running time of the approximation variants with a scaling factor of $t=10$ on near-linear inputs generated in the range of up to $2n$.}
    \label{fig:approxNearLinear}
\end{figure}
\begin{figure}
    \centering
    \includegraphics[width=0.6\linewidth]{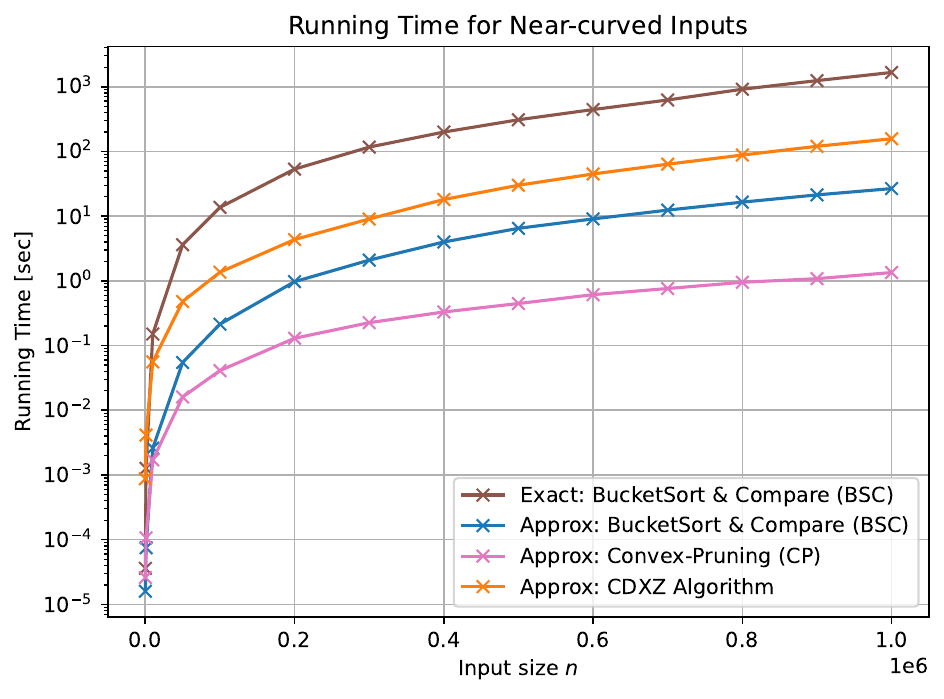}
    \caption{Running time of the approximation variants with a scaling factor of $t=10$ on near-curved generated inputs in the range of up to $2n$.}
    \label{fig:approxNearCurve}
\end{figure}
The performance of CP relies on its ability to prune pairs, which on near-curved instances 
applies to up to 95\% of total point pairs, see Figure \ref{fig:deltaNear}.
\begin{figure}
    \centering
    \includegraphics[width=0.48\linewidth]{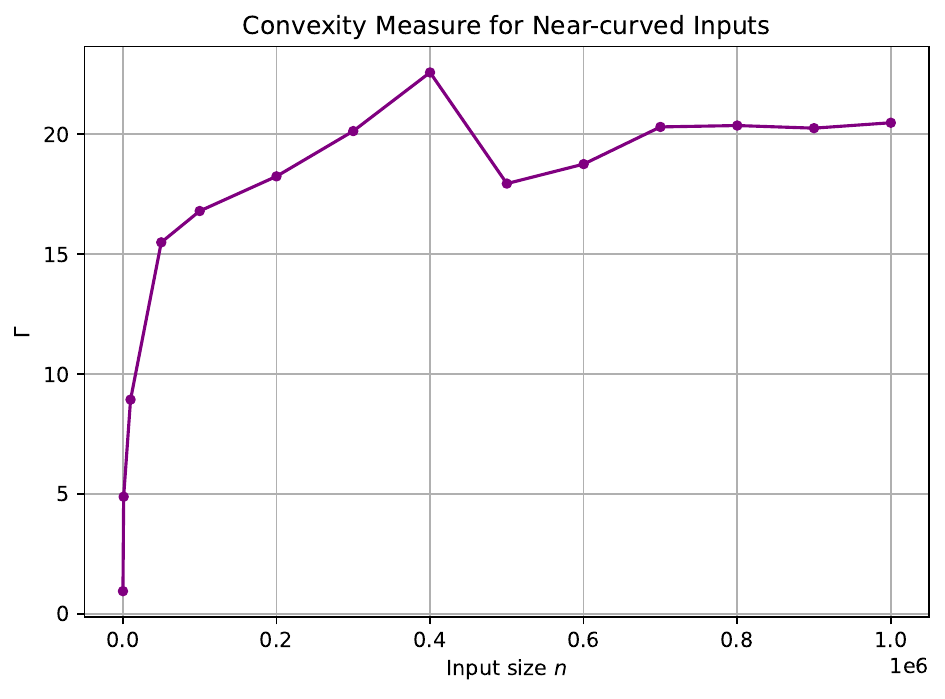} \hfill
    \includegraphics[width=0.48\linewidth]{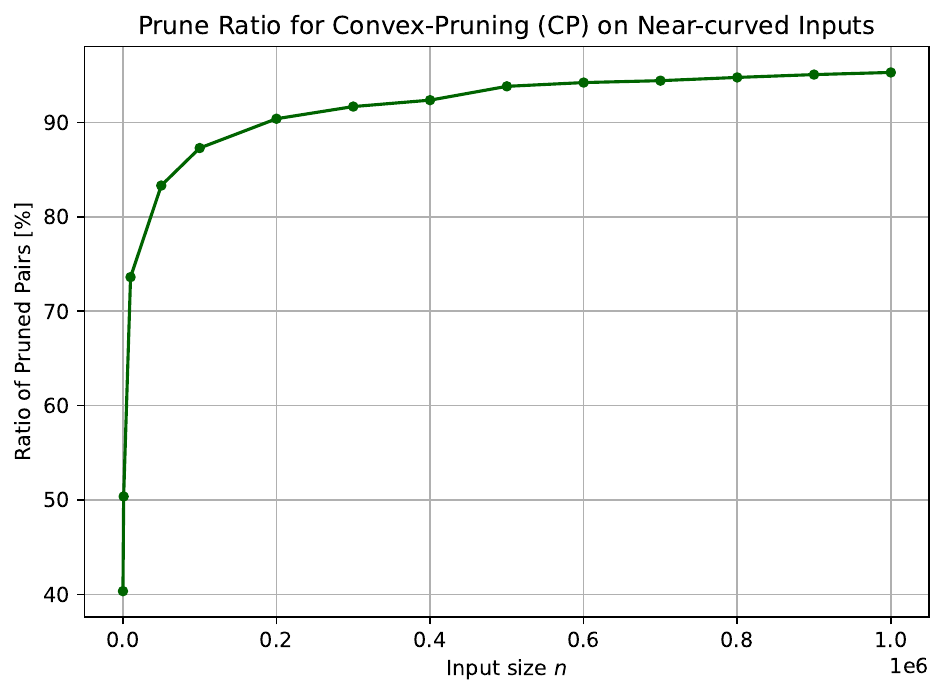}
    \caption{Convexity measure for near-curved inputs (left) and ratios of pruned pairs by the CP algorithm with respect to all available pairs in the index space (right).}
    \label{fig:deltaNear}
\end{figure}

Finally, we consider our third benchmark set, which consists of bi-criteria route planning
instances of varying size.
Figure \ref{fig:approxRealScaling} shows the running time and approximation quality of our algorithms for different scaling factors. We observe that BSC performs best with CP being a close second, especially for larger values of $t$. 
They both improve on exact computation by over an order of magnitude and also stay well below the theoretical additive approximation bound of $2t$.
\begin{figure}
    \centering
    \includegraphics[width=0.7\linewidth]{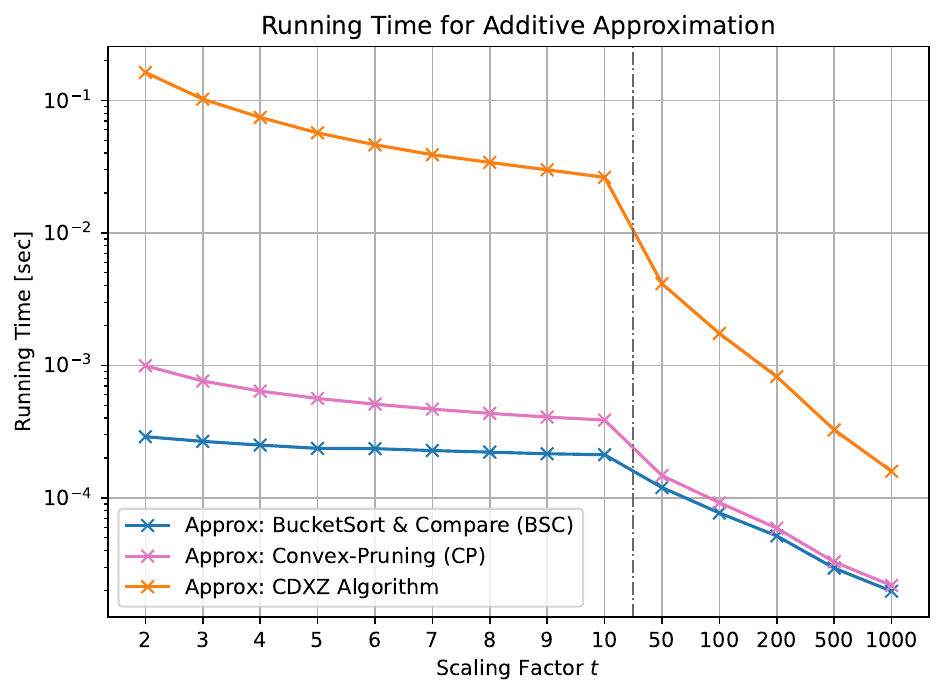}
    \includegraphics[width=0.7\linewidth]{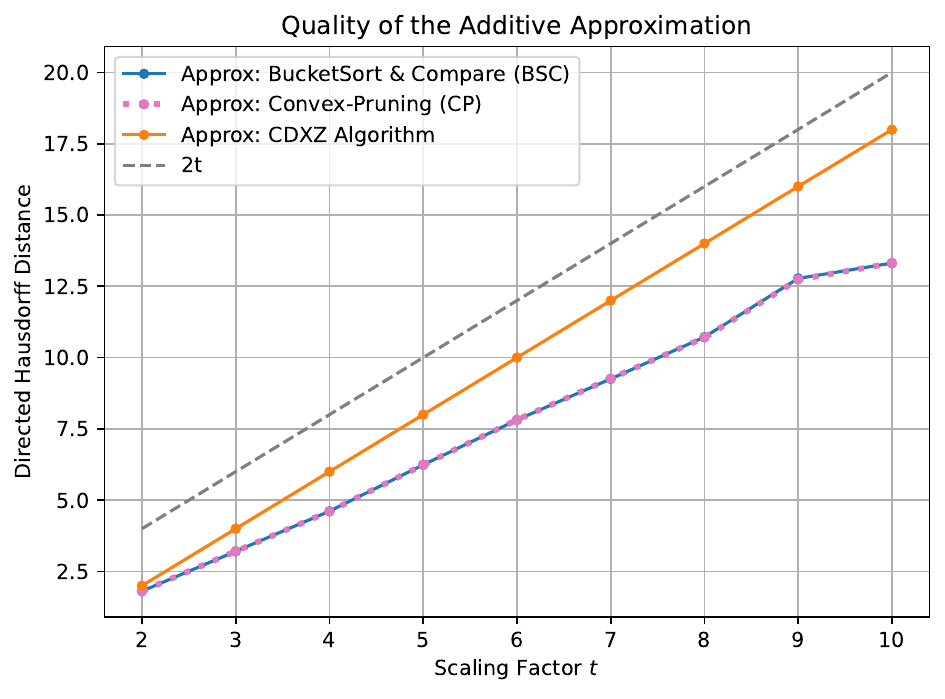}
    \caption{Running time of approximations on real benchmark for varying scaling factors $t=2, \dots, 1000$ (top) and solution quality for $t=2, \dots, 10$ (bottom). This shows the difference of the quality of approximation with (BSC,CP) and without witness reporting(CDXZ).}
    \label{fig:approxRealScaling}
\end{figure}
Examples of approximate solutions can be found in Figures \ref{fig:exampleApprox}.
\begin{figure}
    \centering
    \includegraphics[width=0.48\linewidth]{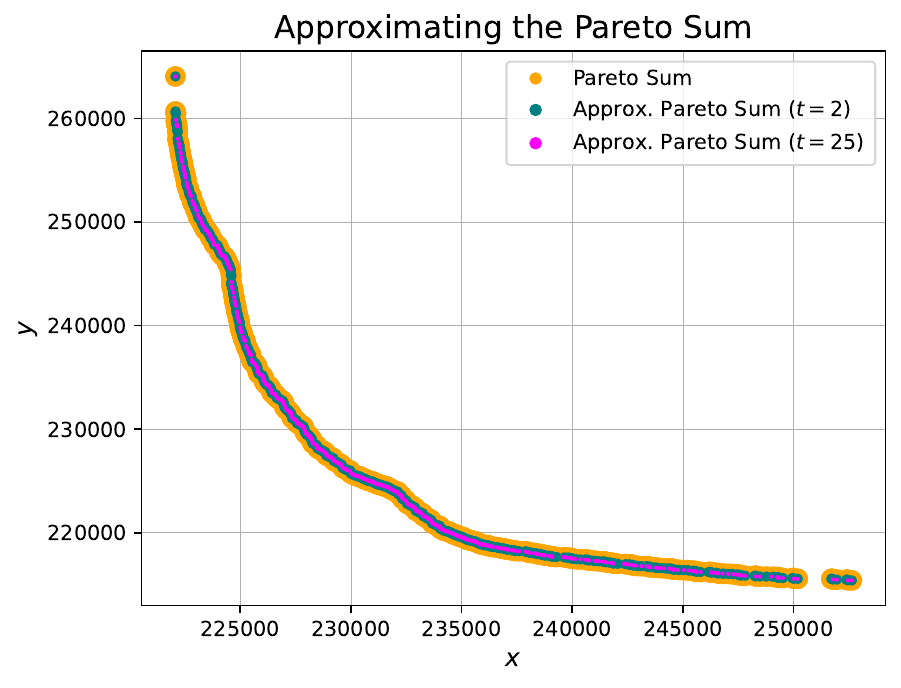} \hfill
    \includegraphics[width=0.48\linewidth]{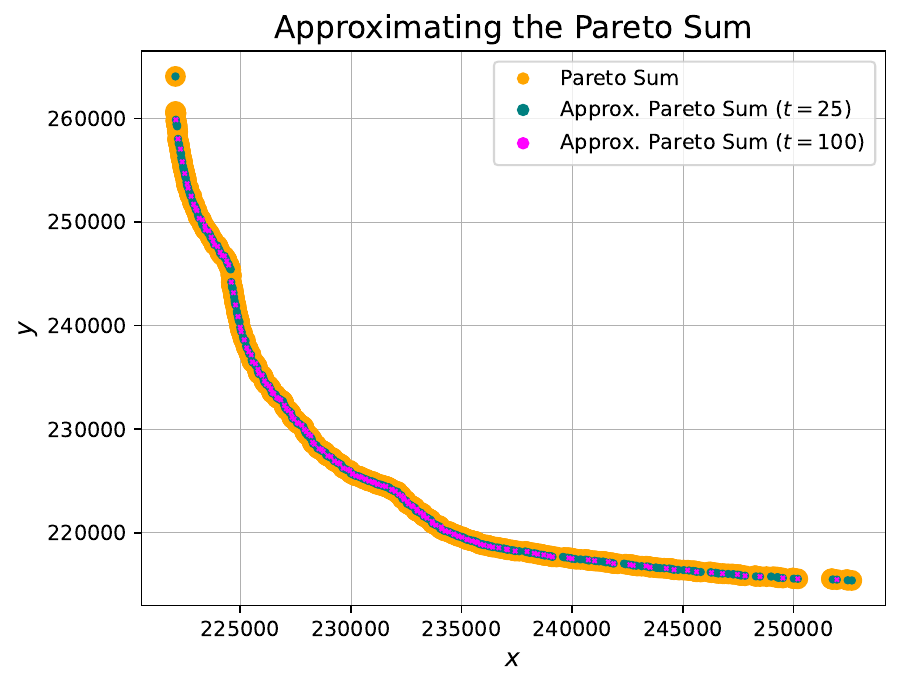}
    \caption{Example for approximate Pareto sum on real benchmark data. The approximative solution for $t=2$ results in a quality of $\Delta = 2$ (left) whereas for $t=25$ the quality is $\Delta = 38$ and for $t=100$ we have $\Delta = 153$ (right).}
    \label{fig:exampleApprox}
\end{figure}
One can see that even when choosing a larger $t$, the  Pareto sum is still well represented by the computed solution. Indeed, larger values of $t$ are quite reasonable in our application. For example, a route that is up to 100 meters longer than another one is almost indistinguishable to a user if the total route length is anyway in the orders of kilometers.

Furthermore, the user usually does not want to be presented with hundreds of routing options, but a concise and diverse set of reasonable options. This is achievable with our approach, as we can not only reduce the running time with larger $t$ but also the output size, see  
Figures \ref{fig:approxVaryT} and \ref{fig:approxOutputSize}.
Clearly, our algorithm provides good trade-offs between running time and solution quality, which allows one to choose the best $t$ depending on the application.

\begin{figure}[h!]
    \centering
    \includegraphics[width=0.5\linewidth]{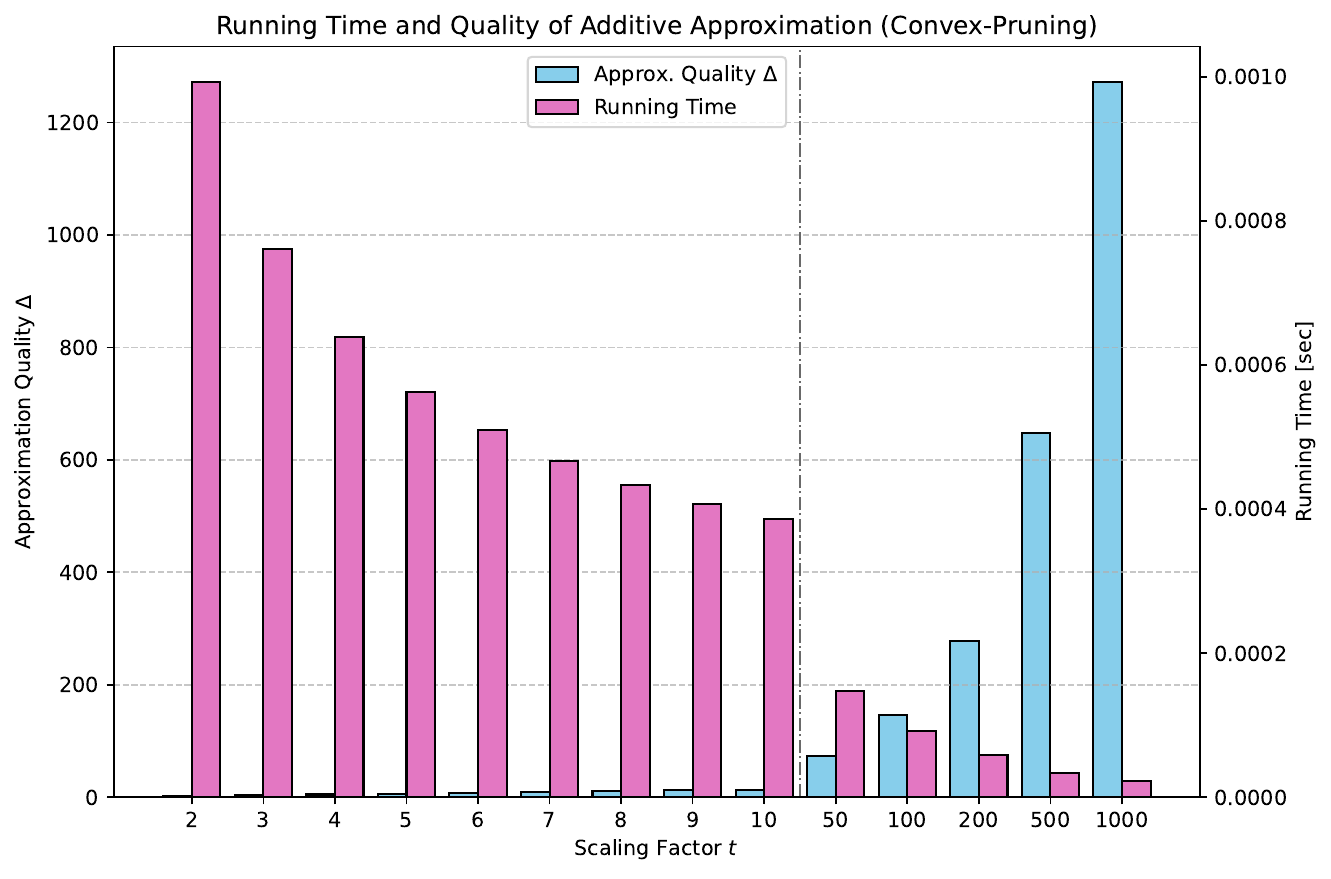}
    \caption{Approximation quality $\Delta$ and running time of the approximation using the CP algorithm on real benchmarks with varying scaling factor $t = 2, \dots, 1000$.}    
    \label{fig:approxVaryT}
\end{figure}

\begin{figure}[h!]
    \centering
    \includegraphics[width=0.5\linewidth]{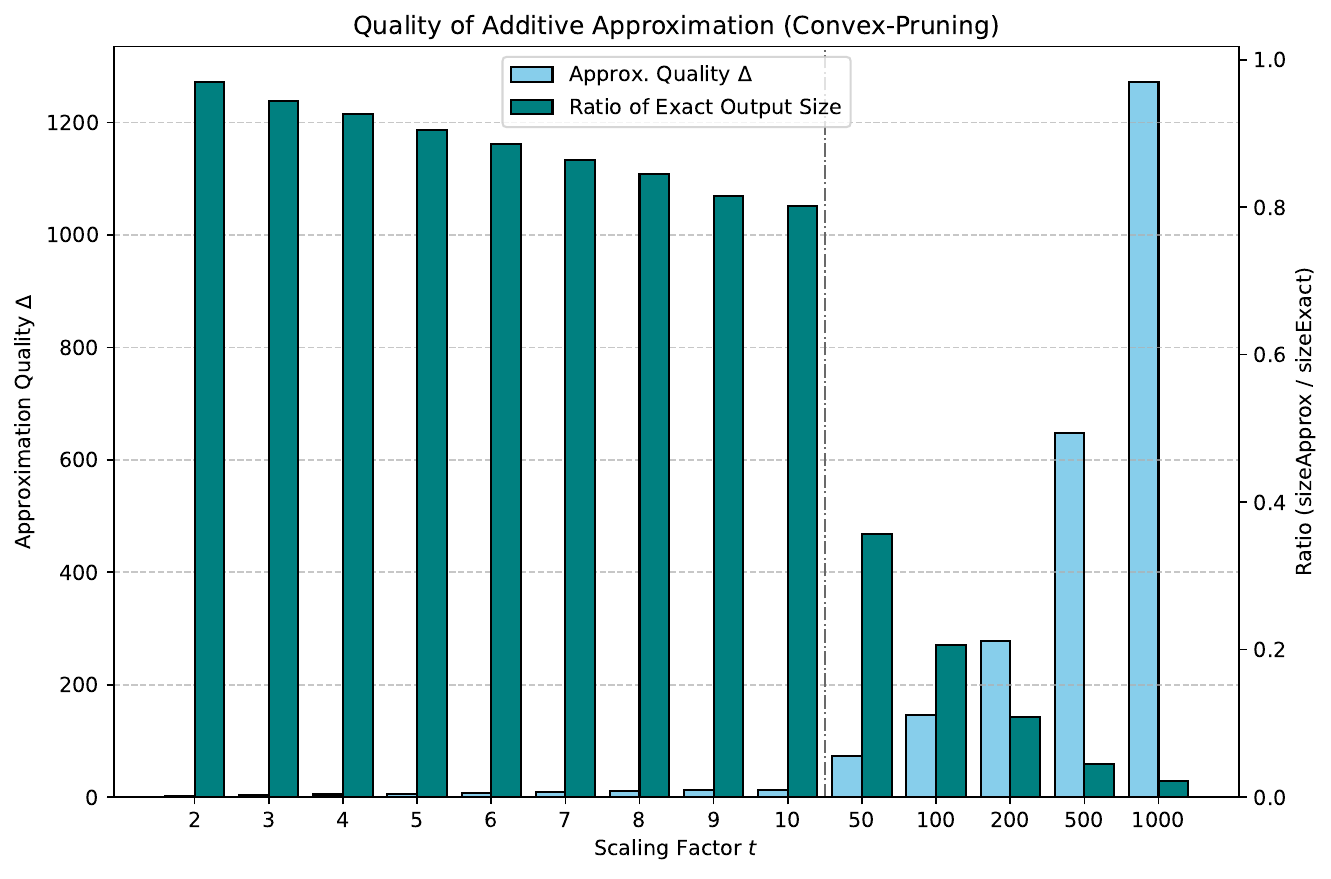}
    \caption{Quality of the approximation using the CP algorithm on real benchmarks for varying scaling factors $t=2, \dots, 1000$. We compare the ratio of the output size of the approximation and the exact algorithm to the approximation quality $\Delta$.}    
    \label{fig:approxOutputSize}
\end{figure}
\section{Future Work \& Discussion}
We have established that fine-grained reductions are not only a means to obtain theoretical intractability results, but also a tool to efficiently solve problems via reductions to other problems for which fast algorithms are implemented and engineered.
To further foster (bounded monotone) min-plus convolution as a core primitive, it would be interesting to know whether the $\tilde{O}(n^{1.5})$ algorithm \cite[Section 4.2]{ChiDX022}, admits an efficient implementation. The main drawback right now is that it needs $9$ calls to the monotone bounded min-plus function, and that the constant $b\in [-4,7]$ is much larger than in the algorithm we implemented.  It seems that significant simplifications are needed beforehand. We believe that the convex pruning algorithm can benefit from practical sparse convolution algorithms, especially on near-linear inputs. Is it possible to engineer fast sparse convolution algorithms in practice? 
Finally, can we leverage our min-plus convolution results to obtain practical improvements for knapsack \cite{BringmannC22,BringmannD024} or the tree sparsity problem \cite{MuchaW019,BackursIS17}? 
\bibliography{lipics-v2021-sample-article}

\appendix
\section{Retracing the steps of the $\tilde{O}(n^{1.6})$-time bounded monotone min-plus convolution algorithm~\cite{ChiDX022}} \label{sec:Chi_Algorithm}
The main idea of the algorithm of Chi, Duan, Xie and Zhang \cite[Section 4.1]{ChiDX022} (in the following called CDXZ algorithm) can be described in two main phases. In the first phase, we compute an approximation of the min-plus convolution of the arrays $A$ and $B$. In particular they compute the min-plus convolution of scaled-down versions of $A$ and $B$, which we denote as $A'$ and $B'$. The values are scaled-down by a factor $n^\alpha$. The benefit of scaling down is that many values in the original arrays $A$ and $B$, will collapse to a same scaled down value. The min-plus convolution of the scaled down version $A',B'$ can now be computed faster, as we will see in Section \ref{sec:Approx_computation}. Let $C'$ denote the min-plus convolution of $A'$ and $B'$.
After computing the min-plus convolution $C'$, scaling up again will not suffice. Scaling the result array $C'$ up, introduces a certain error in each entry $C'[k]$.
In the second phase, we attempt to correct this error. As a main tool, we make use of FFT for fast polynomial multiplication. Lemma \ref{lem:approx}, introduces an interesting connection between the scaled down min-plus convolution and the original min-plus convolution. Every solution $(i,k-i)$ in the original min-plus convolution i.e., fulfilling $A[i]+B[k-i]=C[k]$, will almost correspond to a solution in the scaled-down version. More precisely there exists a $b \in \{0,1\}$ such that  $A'[i] +B'[k-i]= C'[k]+b$. In particular, given a solution $(i,k-i)$ in the scaled-down version, we can now decide if $(i,k-i)$ can potentially be a solution in the original version.
This insight leads to a first idea. By creating polynomials 
\[P(x,y,z):=\sum_{i \in [n-1]} x^{A[i] -n^\alpha A'[i]}y^{A'[i]}z^{i}\] and \[Q(x,y,z):= \sum_{i \in [n-1]}x^{B[i] - n^\alpha B'[i]}y^{B'[i]}z^{i},\] and computing the polynomial multiplication $T(x,y,z):=P\cdot Q$, by looking at the monomials in $T$, we can infer at each index $k=i+j$ (inferred by the degree of $z$), the error introduced by scaling down (inferred by the degree of $x$), and whether $i$ and $j$ are good indices i.e., $(i,j)$ correspond to solutions in the original min-plus convolution problem (inferred by the degree of $y$).
Sadly, this first idea will not work. Calculating the polynomial multiplication of $P$ and $Q$ will cost too much time, and the polynomials $P$ and $Q$ are too large. To reduce the size of the polynomials, we make use of hashing\footnote{The same procedure has been recently used to solve $3$-SUM in preprocessed universes using FFT \cite{Kasliwal0S25}.} modulo a prime $p$ . Instead of noting down the scaled-down values in the polynomials, they first compute this value modulo $p$. Formally they look at the polynomials,
$P(x,y,z):=\sum_{i\in [n-1]} x^{A[i] -n^\alpha A'[i]}y^{A'[i] \mod p}z^{i}$ and $Q(x,y,z) := \sum_{i\in [n-1]} x^{B[i] - n^\alpha B'[i] \mod p}y^{B'[i]}z^{i}$. 
While this again introduces the new problem of pseudo-witnesses i.e., pairs $(i,k-i)$, which satisfy $A'[i] +B'[k-i]  \equiv C'[k] +b \mod p$ but not $A'[i] +B'[k-i] \neq C'[k] +b$, it turns out that we can efficiently identify those and remove them.
When retracing the steps of \cite{ChiDX022} we emphasize the changes we made. 
\subsection{Computing the approximation} \label{sec:Approx_computation}
In order to compute the approximation, we first define the scaled-down versions of the arrays. For the value $n^{\alpha}$, where $\alpha \in (0,1)$ will be defined later, we set $A'[i]:=\left \lfloor \frac{A[i]}{n^\alpha} \right \rfloor$ and $B'[i]:= \left \lfloor \frac{B[i]}{n^\alpha} \right \rfloor $ for every $i \in [n-1].$ As the values in $A$ and $B$ were bounded by $\mathcal{O}(n)$, we can divide $A'$ and $B'$ into $\mathcal{O}(n^{1-\alpha})$ many groups of the same value. Formally, we separate the index space $[n-1]:=[0,a_1-1] \cup [a_1,a_2-1] \dots \cup [a_{m-1},n-1] $, where $a_0:=0, a_m=n$  and all values inside a group in the array $A'$ are equal, that is, for each $j \in[m-1]$ we have $A'[a_j] = \dots =A'[a_{j+1}-1]$. Similarly, for the array $B'$, we separate the index space $[n-1]:= [0,b_1-1] \cup \dots \cup [b_{\ell-1},n-1]$, where $b_0:=0$ and $b_\ell=n$ with the property that all values in $B'$ inside the intervals are equal, that is, for each $j \in [\ell-1]$ we have $B'[b_j]=B'[b_j+1]= \dots = B'[b_{j+1}-1]$.

To compute the min-plus convolution of $A'$ and $B'$, the authors in \cite{ChiDX022} introduce a segment tree with interval updates, where for each $i \in[m-1]$ and $j \in [\ell-1]$, they compute the value $A'[a_i] +B'[b_j]$ and update the minimal value on the interval $[a_i+b_j, a_{i+1}+b_{j+1}-2]$.
This costs time $\mathcal{O}(n^{2-2\alpha} \log(n))$.

We observe that we can remove the segment tree entirely, and thereby the logarithmic factor, by simply making use of the enhanced min-plus convolution Algorithm \ref{alg:enhanced_min_plus}. This will lead to a better runtime of $\mathcal{O}(n^{2-2\alpha}).$

\begin{algorithm}
\caption{Enhanced Min-Plus Convolution $(A,B).$}
\label{alg:enhanced_min_plus}
\begin{algorithmic}
\STATE{$C:=[\infty, \dots, \infty ]$}
\STATE{Divide $A$ into intervals $[0=:a_0, a_1-1], \dots, [a_{m-1},a_m:=n-1]$ and $B$ into intervals $[0=:b_0, b_1-1], \dots, [b_{\ell-1}, b_\ell:=n-1]$ of the same value. }
\FOR{$i \in [m-1]$ }
\FOR{$j \in [\ell-1]$}
\STATE{$C[a_i+b_j] \gets \min \{C[a_i+b_j], A[a_i] + B[b_j] \}$}
\ENDFOR
\ENDFOR

\FOR{$i \in [1,2n-2]$}
\STATE{$C[i] \gets \min \{C[i], C[i-1] \}$ }
\ENDFOR

\RETURN $C$
\end{algorithmic}
\end{algorithm}

We first show a simple observation, which will be beneficial to show the correctness of  the enhanced min-plus convolution algorithms.
\begin{lemma}
If $C$ is the resulting min-plus convolution of monotonically decreasing arrays $A$ and $B$, then $C$ is also monotonically decreasing. \label{lem: decreasing}
\end{lemma}
\begin{proof}
Consider an index $k \in [2n-3]$, then there exists an index pair $(i,(k-i)) \in [n-1] \times [n-1]$, with $C[k] =A[i] +B[k-i]$. We will show that $C[k+1] \leq C[k]$. We give a potential witness for the index $k+1$, which is already not greater than $C[k]$. Either the index pair $(i,(k-i)+1) \in [n-1] \times [n-1]$ or $(i+1,(k-i)) \in [n-1] \times [n-1]$. Assume the first case (the second case is analogous), then by monotonicity of $B$, we have $B[k-i] \geq B[(k-i)+1]$, and thereby $C[k+1] \leq A[i] + B[(k-i)+1] \leq A[i]+B[k-i] \leq C[k].$
\end{proof}
We now proceed with the correctness of Algorithm \ref{alg:enhanced_min_plus}.
\begin{lemma}
Algorithm \ref{alg:enhanced_min_plus} computes the correct monotone min-plus convolution of $A$ and $B$.
\end{lemma}
\begin{proof}
After the division of $A$ and $B$ into intervals, we can assume for each $i \in [m-1 ]$ that $A[a_i]=A[a_i+1]=\dots =A[a_{i+1}-1]$ and for each $j \in [\ell-1]$ that $B[b_j]=B[b_j+1]=\dots = B[b_{j+1}-1]$.

Let $k \in [2n-2]$ be an arbitrary index. We show that the value of $C[k]$ is computed correctly. Denote by $c^*$ the optimal value at index $k$. Clearly $C[k] \geq c^*$.
Let $i \in [n-1]$ be a witness for $k$ i.e. $A[i]+B[k-i]=c^*$. Then $i \in [a_m,a_{m+1}-1]$, and $k-i \in [b_o, b_{o+1}-1]$. We first argue that $C[a_m+b_o]=c^*$. Clearly, $A[a_m]+B[b_o] \leq c^*$ by the computation of the minimum. By Lemma \ref{lem: decreasing}, $C$ is monotonically decreasing, leading to $A[a_m]+B[b_0] \geq c^*$.  Finally, in the last for loop, the propagation of the minimum leads to the fact that $C[i+j] \leq c^*$ as $a_m + b_o \leq i +j \leq a_{m+1}-1 + b_{o+1}-1.$ 
\end{proof}
Dividing $A$ and $B$ into intervals of same value can be done in time $\mathcal{O}(n)$ as they are monotonically decreasing.
In total, Algorithm \ref{alg:enhanced_min_plus} runs in time $\mathcal{O}(m\ell)$, where $m$ and $\ell$ denote the number of groups with same entries in $A$ and $B$ respectively. In the case of the arrays $A'$ and $B'$, we have at most $n^{1-\alpha}$ many groups, resulting in a total runtime of $O(n^{2-2\alpha}).$
After computing the approximation, we can show the following relationship between the scaled-down min-plus convolution and the original min-plus convolution problem. The bound we show in the following will be a bit tighter than the bound calculated by \cite{ChiDX022}.
\begin{lemma}
\label{lem:approx}
If $A[i] + B[j] = C[i+j]$ then $A'[i] +B'[j] = C'[(i+j)] +b $ with $b \in \{0,1\}. $
\end{lemma}
\begin{proof}
We show this by the following calculation. Assume $(i+j)  = (k+l) $ and $C'[(i+j)] = A'[k] + B'[l].$ Then :
\begin{align}
n^{\alpha} \cdot C'[i+j ]  &= n^{\alpha}  \cdot C'[k+l ] \\
                                        & = n^{\alpha}  \cdot  \left( A'[k]  + B'[l] \right)\\
                                        & = n^{\alpha} \cdot  A'[k]  + n^{\alpha}  \cdot  B'[l] \\
                                        & >  A[k] + B[l] - 2 n^{\alpha} \label{eq:greater}  \\
                                        & \geq C[i+j] -2 n^{\alpha}  \\
                                        & = A[i] +B[j] - 2n^{\alpha}  \\
                                        & \geq  n^{\alpha}  \cdot A'[i] +  n^{\alpha}  \cdot B'[j] -2n^{\alpha} 
\end{align}
Thus, by dividing  by $n^{\alpha} $, we obtain $A'[i] +B'[j] - C'[i+j] < 2. $
\end{proof}

The authors in \cite{ChiDX022} have bounded this constant $b \in \{0,1,2\}$, but by observing the strict inequality at Equation \ref{eq:greater}, we can infer that even $b \in \{0,1\}$ holds.
\subsection{Error Correction}\label{sec:error_correction} This section entirely follows \cite{ChiDX022}, there were no major simplifications made by us.
Let $p$ be a prime number, with $n^\beta\leq p\leq 2n^\beta$, where $\beta \in (0,1)$ will be chosen later.
We compute two polynomials \[P(x,y,z):=\sum_{i\in [n-1]}x^{A[i] -n^\alpha A'[i]}y^{A'[i] \mod p}z^{i},Q(x,y,z) := \sum_{i\in [n-1]}x^{B[i] - n^\alpha B'[i] }y^{B'[i] \mod p}z^{i},\]
and denote the polynomial multiplication by $T(x,y,z):=P \cdot Q.$ The polynomials $P$ and $Q$ are bounded by the length $\mathcal{O}(n^{\alpha} \cdot n^\beta \cdot n)=\mathcal{O}(n^{\alpha + \beta +1 })$ and the multiplication costs $\mathcal{O}(n^{\alpha + \beta +1 } \log(n))$ time using FFT.
Furthermore, we denote the set of pseudo-witnesses for $b \in \{0,1\}$ as 
\[T_b:=\{(i,k-i)  \mid  A'[i]+B'[k-i] \equiv C'[k] \mod p  \land A'[i]+B'[k-i] \neq C'[k]]\}\]
\begin{lemma}[\cite{ChiDX022}]
The set $T_b$ has expected size $\mathcal{O}(n^{2-\beta})$ and can be computed in time $\mathcal{O}(|T_b|+n^{2-2\alpha} \log(n))$.
\label{lem:Tb_size}
\end{lemma}
\begin{proof}
Make use of the separation of $A'$ and $B'$ into intervals with the same value. When fixing two intervals $[a_s,a_s-1]$ and $[b_t, b_{t+1}-1]$, the value $A'[i] +B'[j]=\Delta$ is fixed for all $(i,j) \in [a_s,a_s-1] \times [b_t, b_{t+1}]$. The aim is to compute all $(i,j) \in T_b$. After fixing two intervals search for all indices $k \in [a_s+b_t, a_{s-1}+b_{t-1}-2] $ with the property that $C'[k]+b \equiv \Delta \mod p$ and $C'[k]+b \neq \Delta$,
and simply enumerate all $(i,k-i)$ satisfying $ \max\{a_s,k+1-b_{t+1}\} \leq i \leq \min \{ a_{s+1}-1, k-b_t\}$. The indices $k$ can be searched using binary search trees.
This results in a runtime of $\mathcal{O}(n^{2-2\alpha} \log n+|T_b|).$
For any pair $(i,k-i)$ such that $A'[i] +B'[i] \neq C'[k]+b $ as $p$ is chosen uniformly random in the range $[n^\beta,n^{2\beta}]$, it holds that the probability that $A'[i]+B'[i] -C'[k]-b$ is divisible by $p$, is bounded by $\mathcal{O}(1 /\beta).$ By linearity of expectation $E(|T_b|) \leq \mathcal{O}(n^{2-\beta}).$
\end{proof}
Compute for each index $k \in [2n-2]$, and $b \in\{0,1\}$ the polynomial \[R_{p,b}[k](x):= \sum_{(i,k-i) \in T_b} x^{A[i]-n^\alpha A'[i] + B[k-i] - n^\alpha B'[k-i]},\]
and the polynomial
\begin{align*}
C_{p,b}[k](x) := \sum_{} \lambda x^c, &\text{ where }  \lambda x^cy^{C'[k]+b \mod p} z^k  
                                  \text{ is a monomial in } T(x,y,z). 
\end{align*}
Finally, by computing the value  \begin{equation}\label{eq:sbk}
    s_{b}[k]:= \min \{c: \lambda x^c \text{ is a monomial in } C_{p,b}[k]-R_{p,b}[k] \}
\end{equation} for each $b \in \{0,1\}$ and $k \in [2n-2],$
we can reconstruct the value of $C[k]$ as  \[C[k]= \min_{b \in \{0,1\}} \left \{n^\alpha (C'[k]+b) +s_{b}[k]\right \}.\]
The correctness follows by the same calculation as performed by \cite{ChiDX022}, but intuitively follows from the fact that in $s_b[k]$ we subtract from the total number of witnesses the number of pseudowitnesses. Although \cite{ChiDX022} have shown correctness for the case of monotonically increasing arrays $A$ and $B$, the proof is indifferent of this fact, and works the same for monotonically decreasing arrays $A$ and $B$.
The final runtime will be $\mathcal{O}(n^{2-2\alpha}) + \mathcal{O}(n^{2-\beta} \log{(n)})+\mathcal{O}(n^{1+\alpha+\beta} \log(n))$. For the choice of $\alpha=0.2$ and $\beta =0.4$, the runtime will amount to $\mathcal{O}(n^{1.6}\log{n}).$

\section{Convex-Pruning Algorithm \cite{BringmannC23}} \label{sec:convex_pruning}
We summarize the main idea of the algorithm of \cite{BringmannC23}, for the details we refer to their highly readable paper.
The algorithm of Bringmann and Cassis \cite{BringmannC23} works as follows:
A function $f':[n] \to \mathbb{Z}$ is a $\Gamma$-near convex approximation of a function $f:[n] \to \mathbb{Z}$, if $f'$ is convex and $f'(i) \leq f(i) \leq f'(i)+ \Gamma$ holds for each $i\in[n-1]$ i.e. $f'$ additively approximates the function $f$ by a factor of $\Gamma$. Let $A$ and $B$, be the input arrays, we want to compute the min-plus convolution of $A$ and $B$.

We first construct a convex approximation $A'$, where $A'[i] \in \mathbb{Q}$, of the array $A$. The way we do this is by computing the lower convex hull (by a Graham scan for instance), of the points $A= \{(i, A[i]): i \in [n-1] \}$. We perform a linear interpolation between all the points in the convex hull to get a valid point for each index $i \in [n-1]$. The same is performed for the array B to obtain a convex approximation $B'$. 
It is simple to see that this is the best point-wise convex approximation of the points in $A$ and $B$. Moreover, we get values $\Gamma_A$ and $\Gamma_B$, which we can compute by computing $ \Gamma_A:=\max_{i \in [n-1]} A[i]-A'[i]$. Similarly for the array $B$, we can compute $\Gamma_B$. We set $\Gamma= \max\{1,\Gamma_A, \Gamma_B \}.$

Compute the min plus convolution of the convex arrays $A'$ and $B'$ and store it in an array $C'$. This can be done in linear time by known methods such as the SMAWK algorithm by the authors in \cite{Aggarwal86}, which also computes  witnesses for $C'$ i.e. for each $k \in [2n-2]$, a witness $i \in [n-1]$ such that $A'[i] + B'[k-i] = C'[k].$

We now concentrate on finding the min-plus convolution of the original arrays $A$ and $B$, it turns out that the convex approximations can help us in pruning the search space of all possible indices $(i,j) \in [n-1] \times [n-1].$

For this, the \emph{relevant} region $R_{2\Gamma}:= \{(i,k-i): A'[i] + B'[k-i] \leq C'[i+j] + 2\Gamma\}$ is introduced.
Let us now show the key lemma for the convex pruning algorithm.
\begin{lemma}[\cite{BringmannC23}]
If $(i,j) \not \in R_{2\Gamma}$ then $A[i] + B[j] > C[i+j].$
\end{lemma}
\begin{proof}
Assume $(i,j) \not \in R_{2\Gamma}$, and that $C'[i+j] = A'[x] + B'[y]$ with $x+y=i+j$, then
\begin{align*}
A[i]+ B[j] \geq A'[i] + B'[j] > C'[i+j]+ 2 \Gamma &= A'[x] + \Gamma + B'[y] + \Gamma \\ 
&\geq A[x]+B[y] \geq C[i+j].
\end{align*}
\end{proof}
Thus, all pairs $(i,j) \in [n-1]^2$, which are not in the relevant region, do not need to be checked in order to compute the min-plus convolution of $A$ and $B$.  We will now see how one can use this idea to prune larger segments $(I,J) \subseteq [n-1]^2$.

First, we visualize the index space $ [n-1] \times [n-1]$ by a square, with the point $(0,0)$ being in the lower left corner and $(n,n)$ in the upper right corner; see Figure \ref{fig:near_convex}.
As a final definition, we need that of a monotone connected path.
We call a set $P\subseteq [n-1]\times [n-1]$ a connected monotone path\footnote{In Figure \ref{fig:near_convex} a monotone path crosses each diagonal $i_k+j_k=k$ precisely once.} if for every $k \in [2n-2]$, we have precisely one $(i_k,j_k)$ such that $i_k+j_k=k$ and that if $(i_k,j_k) \in P$ then either $(i_k+1,j_k) \in P$ or $(i_k,j_k+1) \in P.$
Bringmann and Cassis are able to show that the relevant region $R_{2\Gamma}$ is structured. The region $R_{2\Delta}$ is bounded between two monotone connected paths; see Figure \ref{fig:near_convex}.

Lastly, we can decide if a point $(i,j)$ is above or below the relevant region, by simply checking if we are outside the relevant region, and use the witness path for the convex approximations as a reference point. It even suffices to check one corner for a square $[i_A , i_B] \times [j_A, j_B]$; see Figure \ref{fig:near_convex}.

After computing the convex approximations the value $\Gamma$ and the witness path for the convex approximations, we can perform a call to Algorithm \ref{alg:convex_pruning} (which was also described\footnote{We remark that the pseudocode of Bringmann and Cassis has minor typos.} in \cite{BringmannC23}), with parameters $I=J=[n-1]$, the arrays $A,B$ for which we want to compute the min-plus convolution and an array $C=[\infty, \dots, \infty]$ of length $2n-2$, which we will improve throughout the algorithm. In the end $C$ will be the min-plus convolution of the arrays $A$ and $B$.

\begin{algorithm}
\caption{Convex-Pruning Algorithm to compute the min-plus convolution of $A$ and $B$ \cite{BringmannC23}.} 
\label{alg:convex_pruning}
\begin{algorithmic}
\STATE{\textbf{convex pruning}$(I=[i_A,i_B], J=[j_A,j_B],A,B, C):$}
\STATE{\quad \textbf{if }$(i_A,j_B) \not \in R_{2\Gamma} $ and $(i_A,j_B)$ \text{is above witness path} \textbf{then return} }
\STATE{\quad \textbf{elif }$(i_B,j_A) \not \in R_{2\Gamma} $ and $(i_B,j_A)$ \text{is below witness path} \textbf{then return} }
\STATE{\quad \textbf{elif} $(i_A,j_B) \in R_{2\Gamma}$ and $(i_B,j_A) \in R_{2\Gamma}$  \textbf{then}}
\STATE{\quad \quad $ \tilde{C}[i_A+j_A, \dots, i_B+j_B]\gets \text{min-plus convolution}(A[i_A,\dots, i_B],B[j_A, \dots, j_B ])$}
\STATE{\quad \quad \textbf{for $k \in [i_A+j_A, i_B+j_B]$ do}}
\STATE{\quad \quad \quad $C[k] \gets \min(C[k], \tilde{C}[k])$}
\STATE{\quad \quad \textbf{end for}}
\STATE{\quad \textbf{else} split $I$ and $J$ into two halves $I_1,I_2,J_1,J_2$, and recursively compute convex pruning on all 4 combinations}
\end{algorithmic}
\end{algorithm}

The theoretical (and practical) runtime of the algorithm heavily depends on which min-plus convolution, we plug in in the computation of $\tilde{C}$ in Algorithm \ref{alg:convex_pruning}.
Bringmann and Cassis observed that a sparse convolution is well suited in this case, as the size of the sumset $\{(i,A[i]) \mid i \in [i_A,i_B]\} + \{(j,B[j]) \mid j \in [j_A,j_B] \}$ is small. Using sparse convolutions in the convex pruning algorithm will result in Bringmann and Cassis delta-near convex min-plus convolution algorithm and run in time $\tilde{O}(\Gamma n).$ Due to the lack of practical and efficient sparse-convolution algorithms, we instead inserted the fastest min-plus convolution variant we had, which was the enhanced min-plus convolution algorithm described in Algorithm \ref{alg:enhanced_min_plus}. Notice that this is a degree of freedom and one can also plug in other min-plus convolution computations in there, such as the CDXZ algorithm. Plugging in the enhanced min-plus convolution algorithm in the computation of $\tilde{C}$ in Algorithm \ref{alg:convex_pruning} leads to a runtime of $\Theta(W^2)$, where $W$ denotes the length of the arrays $A$ and $B$.
Nevertheless, there is still potential for the algorithm to work well in practice, due to the pruning. Plugging the enhanced min-plus convolution algorithm in the convex pruning algorithm and then into our Pareto Sum framework leads to a runtime of $O(n+W^2).$
We remark that the algorithm also works for differing sizes of $A$ and $B$, for the runtime analysis and correctness proof, we refer to \cite{BringmannC23}.
\begin{figure}[h!]  % or [t], [b], [H], etc.
    \centering
\begin{tikzpicture}[scale=0.5]

% ================= LEFT FIGURE ===================

% Boundary square
\draw[thick] (0,0) rectangle (6,6);

% First path (blue)
\draw[blue, very thick]
    (0,0) -- (1,0) -- (1,1) -- (2,1) -- (3,1) -- (3,2) --
    (4,2) -- (4,3) -- (5,3) -- (5,4) -- (6,4) -- (6,6);

% Second path (red)
\draw[red, very thick]
    (0,0) -- (0,1) -- (1,1) -- (1,2) -- (2,2) -- (2,3) --
    (3,3) -- (3,4) -- (4,4) -- (4,5) -- (5,5) -- (5,6) -- (6,6);

% Shaded region
\begin{scope}
    \clip (0,0)
        -- (1,0) -- (1,1) -- (2,1) -- (3,1) -- (3,2)
        -- (4,2) -- (4,3) -- (5,3) -- (5,4) -- (6,4) -- (6,6)
        -- (5,6) -- (5,5) -- (4,5) -- (4,4) -- (3,4)
        -- (3,3) -- (2,3) -- (2,2) -- (1,2) -- (1,1)
        -- (0,1) -- (0,0) -- cycle;
    \fill[blue!15] (0,0) rectangle (6,6);
\end{scope}

% Labels
\node at (3,2.7) {$R_{2\Gamma}$};
\node at (-0.3,-0.3) {$(0,0)$};
\node at (6.4,6.4) {$(n,n)$};

% ================= RIGHT FIGURE ===================

\begin{scope}[shift={(8,0)}]

% Boundary square
\draw[thick] (0,0) rectangle (6,6);

% First path (blue)
\draw[blue, very thick]
    (0,0) -- (1,0) -- (1,1) -- (2,1) -- (3,1) -- (3,2) --
    (4,2) -- (4,3) -- (5,3) -- (5,4) -- (6,4) -- (6,6);

% Second path (red)
\draw[red, very thick]
    (0,0) -- (0,1) -- (1,1) -- (1,2) -- (2,2) -- (2,3) --
    (3,3) -- (3,4) -- (4,4) -- (4,5) -- (5,5) -- (5,6) -- (6,6);

% Shaded region
\begin{scope}
    \clip (0,0)
        -- (1,0) -- (1,1) -- (2,1) -- (3,1) -- (3,2)
        -- (4,2) -- (4,3) -- (5,3) -- (5,4) -- (6,4) -- (6,6)
        -- (5,6) -- (5,5) -- (4,5) -- (4,4) -- (3,4)
        -- (3,3) -- (2,3) -- (2,2) -- (1,2) -- (1,1)
        -- (0,1) -- (0,0) -- cycle;
    \fill[blue!15] (0,0) rectangle (6,6);
\end{scope}

% Upper small square
\draw[black, thick, fill=white]
    (0.5,4.3) rectangle (1.5,5.3);

% Highlight lower-right corner of upper square
\fill (1.5,4.3) circle (3pt);

% Lower small square
\draw[black, thick, fill=white]
    (4.2,0.2) rectangle (5.2,1.2);

% Highlight upper-left corner of lower square
\fill (4.2,1.2) circle (3pt);

% Labels
\node at (3,2.7) {$R_{2\Gamma}$};
\node at (-0.3,-0.3) {$(0,0)$};
\node at (6.4,6.4) {$(n,n)$};

\end{scope}

\end{tikzpicture}
    \caption{To the left, we see that the relevant region $R_{2\Gamma}$ is bounded by the blue and red monotone connected paths. To the right, we see two possible squares, defined by $[i_A, i_B] \times [j_A ,j_B]$. Both squares are in the non-relevant region, and can thus be ignored entirely for the min-plus convolution computation. It suffices to check only one point due to the monotone connected property of the paths. Namely, for the upper square it suffices to check the lower right corner $(i_A,j_B)$ and for the lower square it suffices to check the upper left corner $(i_B,j_A)$. We know in which part of the region we are, by checking if we are below or above the witness path, which always is contained in the relevant region. If a square is relevant it suffices to check the points defined by the upper left and lower right corner for relevance in the algorithm.  }
    \label{fig:near_convex}  % optional
\end{figure}
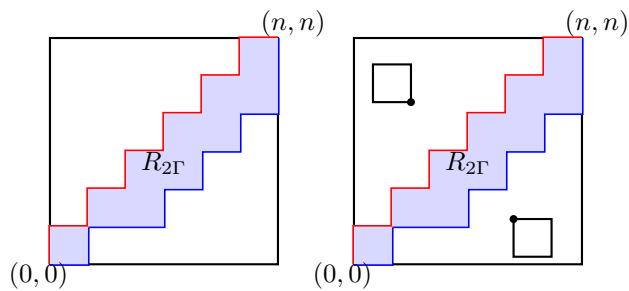

\section{Bucketsort \& Compare}\label{sec:bucketsort_and_compare}
In the bounded integer case, we can improve the sort and compare algorithm from \cite{FunkeHSST24LB} in a very simple manner. The sorting occurs in a typical bucketsort fashion, afterwards it remains to create the pareto set. The final algorithm will be very simple as sketched in Algorithm \ref{alg:bucketsort_and_compare}. Note that the algorithm assumes that the $x$-coordinates are bounded by $W$, and that we already translated the $x$-coordinate points of $P$ and $Q$ such that the smallest $x$-coordinate is 0, we can remedy this easily, by translating back the points in $S$ afterwards.
\begin{algorithm}[h!]
\caption{Bucketsort and compare$(P,Q)$}
\label{alg:bucketsort_and_compare}
\begin{algorithmic}
\STATE{$C:=[\infty, \dots, \infty ]$ }
\FOR{$p \in P$}
\FOR{$ q \in Q$}
\STATE{$C[p.x + q.x] \gets \min \{C[p.x+q.x], p.y+q.y \}$}
\ENDFOR
\ENDFOR
\STATE{$S \gets \emptyset, \ell \gets \infty $}
\FOR{$i = 0$ to $2W$}
\IF{$C[i]< \ell$}
\STATE{$S \gets S \cup \{(i,C[i])\}$}
\STATE{$\ell \gets C[i]$}
\ENDIF
\ENDFOR
\RETURN $S$
\end{algorithmic}
\end{algorithm}

The running time of the algorithm is clearly in $\mathcal{O}(n^2 + W)$. It is simple to extend the algorithm above with witness reporting.
\section{Exact Pareto sum algorithms \cite{FunkeHSST24LB}}  \label{sec:direct_algorithms}
Based on the experiments conducted in \cite{FunkeHSST24LB}, the fastest algorithms for Pareto sum computation across a large variety of inputs are the Sort \& Compare algorithm (SC) and the Successive Sweep Search (SSS) algorithm.

Both algorithms start by sorting the input sets $P$ and $Q$ lexicographically, and then operate on the matrix $M$ where $M_{i,j} = P_i + Q_j$. In the SC algorithm, a min-priority queue (PQ) is initialized with the elements $M_{1,1}$ to $M_{1,n}$. Whenever an element $M_{i,j}$ extracted from the  PQ, it is added to the solution set $S$ if it is not dominated by the point last inserted in $S$. In any case, its column successor in $M$, that is, $M_{i,j+1}$ is inserted next into the PQ unless it is also dominated by the latest element added to $S$. The algorithm was shown to run in time $\mathcal{O}(n^2 \log n)$ using $\mathcal{O}(n)$ space. The SSS algorithm uses a range-min oracle to detect the elements in $S$ one-by-one. $M_{1,1}$ is always part of $S$. Let now $x_{\min}$ and $y_{\max}$ denote the coordinates of the last point added to $S$.
The algorithm starts at $M_{n,1}$, that is, the last entry of the first column. Whenever it enters a new column $j$, it uses upwards linear search in that column until it  reaches an entry $M_{i,j}$ where either $M_{i,j}.x > x_{\min} $ and $M_{i-1,j}.x \leq x_{\min}$ or   where $M_{i,j}.y < y_{\max} $ and $M_{i-1,j}.y \geq y_{\max}$. It then switches to the next column. Among the elements  at which such a switch occurred, it selects the smallest one and adds it to $S$. The process is repeated until the oracle does not find a new point. 
As each call to the range-min oracle can be implemented to take time in $\mathcal{O}(n)$, the algorithms runs in time $\mathcal{O}(nk +n \log n)$ and uses again $\mathcal{O}(n)$ space.
Note that for bounded Pareto sum, we always have $k \in \mathcal{O}(W)$. 

\section{Witness Reporting}\label{sec: witness_reporting}
In this section, we extend the CDXZ algorithm for reporting witnesses. Formally, we want for each $C[k]$, where $k \in [2n-2]$, an index $i$ such that $C[k] = A[i] + B[k-i]$.  Witness reporting has been mostly studied for boolean matrix multiplication \cite{AlonN96, Seidel92} but the techniques transfer to the boolean convolution setting (as is the case in the monotone bounded min-plus convolution algorithm). 
In Algorithm \ref{alg:witness_report}, we introduce a witness finding algorithm for CDXZ using techniques of Alon and Naor  \cite{AlonN96,Seidel92} for boolean matrix multiplication witness finding.

The algorithm is based on the following simple intuition. Assume, we want to compute witness for the usual polynomial multiplication of $C(x) := A(x) \cdot B(x)$, where $A(x)$ and $B(x)$ have coefficients in $\{0,1\}$.
If the coefficient of some monomial $x^k$ in the result polynomial $C(x)$ is $1$, then we can recover its unique witness as follows: define a slight variation $A'$ of $A$ by replacing any monomial $x^{i}$ by $ix^{i}$. In the product polynomial $C'(x) = A'(x) \cdot B(x)$, we can read off a witness $i$ (such that $x^{i}$ is a nonzero monomial in $A$ and $x^{k-i}$ is a nonzero monomial in $B$) as the coefficient of $x^{k}$ in $C'(x)$.

Alon and Naor's approach now works in stages to also find witnesses for monomials $x^{k}$ with many witnesses: each stages gradually sifts down the number of witnesses by (randomly) ignoring/deleting some elements in $B(x)$, and then performing the observation to compute the unique indices.

Our adaption of the CDXZ algorithm will in fact find witnesses (if they exist) for the values $n^\alpha(C'[k] +b) +s_b[k]$ for both values of $b \in \{0,1\}.$ One can then solve the original problem, by looking up which $b$ was the true minimizer of $C[k]$, and choosing that witness.
We crucially make use of the following observation.
\begin{observation}[Witness count] \label{obs: fft_witness}
Let $p$ be a fixed prime number, and $b\in \{0,1\}$.
Assume that $S \subseteq [n-1]$, and we have the polynomials
\[P(x,y,z):= \sum_{i \in [n-1]}  \left( x^{A[i] -n^\alpha A'[i]}y^{A'[i] \bmod p}z^{i} \right), \] 
and the polynomial
\[Q(x,y,z) := \sum_{i\in [n-1]} x^{B[i] - n^\alpha B'[i]}y^{B'[i] \bmod p}z^{i}.\]
After calculating the following univariate polynomials for each index $k$ 
\[C[k](x) := \sum_{} \lambda x^c, \text{ where }  \lambda x^cy^{C'[k]+b \bmod p} z^k  
                                  \text{ is a monomial in } P(x,y,z) \cdot S(x,y,z).\]
and taking the polynomials representing the pseudowitnesses 
\[R[k](x):= \sum_{(i,k-i) \in T_b} \left(x^{A[i]-n^\alpha A'[i] + B[k-i] - n^\alpha B'[k-i]} \right).\]
We can observe that 
\[Z[k] \gets \lambda \text{ where } \lambda \text{ is the coefficient of } x^{s_b[k]} \text{ in }C[k](x) - R_S[k](x),\]
computes the number of true witnesses for the index $k$ (as we subtract the number of pseudowitnesses from the number of total witnesses).

Moreover, if $Z[k] = 1$, then we can compute a unique witness by reperforming the above calculation, but where we multiply each monomial $x^{i}$ in $P(x,y,z)$ and each monomial $c_i x^i$ in $R[k](x)$ by its witness index $i$.
\end{observation}
We will try to gradually sift down the value of true witnesses $Z[k]$ by dropping some elements in the polynomial $Q(x,y,z)$.
Thus, throughout the algorithm, the polynomial $Q(x,y,z)$ and $R_S[k](x)$ will be sifted via a set $S$, by defining 
\[Q_S(x,y,z) := \sum_{i\in S} x^{B[i] - n^\alpha B'[i]}y^{B'[i] \bmod p}z^{i},\] and 
 \[R_S[k](x):= \sum_{(i,k-i) \in T_b, k-i \in S}  x^{A[i]-n^\alpha A'[i] + B[k-i] - n^\alpha B'[k-i]} ,\]where initially $S =[n-1].$
 In the case that at some point throughout the algorithm, we have that $Z[k] = 1$ (which is dependent on $S$), then we can find a witness by recomputing $Z[k]$, but with the changed polynomials
\[P'(x,y,z):= \sum_{i \in [n-1]}  i \cdot \left( x^{A[i] -n^\alpha A'[i]}y^{A'[i] \mod p}z^{i} \right), \] 
and 
\[R'_S[k](x):= \sum_{(i,k-i) \in T_b, k-i \in S} i \cdot \left(x^{A[i]-n^\alpha A'[i] + B[k-i] - n^\alpha B'[k-i]} \right).\] 
In the following, let $c := \lceil \log \log n + 9 \rceil $, $ \alpha := \frac{8}{2^c}$ and $t = \lceil 1+ 3\log_{4/3}(n) \rceil.$
\begin{algorithm}
\caption{Witness reporting for monotone bounded min-plus convolution.}
\label{alg:witness_report}
\begin{algorithmic}
\STATE{Let $p$ be a fixed prime number chosen as in Section \ref{sec:error_correction}.}
\STATE{Compute the bounded monotone min-plus convolution of $A$ and $B$ as described before.}
\FOR{$b \in \{0,1\}$}
\STATE{Let $s_b$ be the resulting array as defined in Equation \ref{eq:sbk}}
\WHILE{there exists an index $k$ with no witness}
\STATE{Let $L$ be the set of indices $k \in [2n-2]$ with no known witness.}
\STATE{$S_1 \gets \{0, \dots n-1\}$}
\STATE{$T(x,y,z) \gets P(x,y,z) \cdot  Q_{S_1}(x,y,z) $}
\FOR{$k \in [2n-2]$}
\STATE{$C_{1}[k](x) := \sum_{} \lambda x^c, \text{ where }  \lambda x^cy^{C'[k]+b \mod p} z^k  
                                  \text{ is a monomial in } T(x,y,z).$ }
\ENDFOR
\FOR{$k \in[2n-2]$}
\STATE{$Z_{1}[k] \gets \lambda \text{ where } \lambda \text{ is the coefficient of } x^{s_b[k]} \text{ in }C_1[k](x) - R_{S_{1}}[k](x)$}
\ENDFOR
\FOR{$ i = 1, \dots, t$} 
\STATE{Let $L'$ denote set of all $k \in [2n-2]$ where $Z_i$ has a non-zero coefficient at most $c$.}
\STATE{Find witnesses for all indices $k$ in $L'$ and mark those in $L$}
\STATE{$S_{i+1} \gets$ \emph{good} set}
\STATE{$T_{i+1}(x,y,z) \gets P(x,y,z) \cdot  Q_{S_{i+1}}(x,y,z) $}
\FOR{$k \in [2n-2]$}
\STATE{$C_{i+1}[k] := \sum_{} \lambda x^c, \text{ where }  \lambda x^cy^{C'[k]+b \mod p} z^k  
                                   \text{ is a monomial in } T(x,y,z).$ }
\ENDFOR
\FOR{$k \in[2n-2]$}
\STATE{$Z_{i+1}[k] \gets \lambda \text{ where } \lambda \text{ is the coefficient of } x^{s_b[k]} \text{ in }C_{i+1}[k](x) - R_{S_{i+1}}[k](x)$}
\ENDFOR

\ENDFOR
\ENDWHILE
\ENDFOR
\RETURN witnesses
\end{algorithmic}
\end{algorithm}
We will closely follow the correctness proof from Alon and Naor \cite{AlonN96}.
We call $S_{i+1} \subseteq S_i$ a good set if the following conditions hold:
\begin{itemize}
    \item The sum of the non-zero entries $Z_{i+1}$ is at most $3/4$-th of the entries in $Z_{i}$ (This leads to the fact that the algorithm terminates in at most $1+3\log_{4/3}(n)$ many steps).
    \item The fraction of entries in $L$ that are $0$ in $Z^{i+1}$ but were not in $L'$ is at most $\alpha.$
\end{itemize}
We find a good set by removing each element in $S_{i+1}$ with probability 1/2, repeating this process until a good set is found. After $O(\log(n))$ iterations, with high probability, we find a good set. Below, we show that such a set exists with a non-zero probability.
Below, we also show how to find witnesses for all indices $k$ in $L'$.
\begin{claim}\label{claim:good}
The probability that the sum of the entries in $Z_{i+1}$ is at most $3/4$-th of the sum of the entries in $Z_i$ is at least $1/3$.
\end{claim}
\begin{proof}
Every entry $i \in S$ is removed with probability $1/2$.  Thus the expected number of true positives also halves. Applying a Markov bound concludes the claim.
\end{proof}

\begin{claim}
The probability that a fixed entry in $Z_i$ which is at least $c$ drops to $0$ in step $i+1$ is at most $1/2^{c}$. 
\end{claim}
\begin{proof}
We assume the worse case that we have independence on the indices and that we only pick true witnesses.
We need to have removed at least $c$ elements from $S$ for the number of true witnesses to go down to 0, which happens with probability $1/2$ each.  This leads to a probability of at most $1/2^{c}$
\end{proof}

\begin{claim}\label{claim:bad}
The probability that more than a fraction of $ \alpha$, in $Z_i$ which had value at least $c$ drop down to $0$ in $Z_{i+1}$ in step $i+1$, is at most $2^{-c} \alpha^{-1} = 1/8.$
\end{claim}
\begin{proof}
This follows by the probability above and applying a Markov bound.
\end{proof}

\begin{lemma}
During the algorithm $S_{i+1}$ is good with probability at least 1/6.
\end{lemma}
\begin{proof}
We take the probability of the good event in Claim \ref{claim:good} and the bad event in Claim \ref{claim:bad}, and subtract that is $1/3- 1/8 > 1/6.$
\end{proof}

\begin{lemma}\label{lem: pos_prob}
For each index $k$ in the entry $L'$ there is a positive probability that $Z_{i+1}[k]=1.$
\end{lemma}
\begin{proof}
This event depends on at most $c$ variables and is at least $c/2^c = \Theta(\frac{\log \log(n)}{ \log(n)})$.
\end{proof}

\begin{lemma}
The algorithm is sound i.e. if the algorithm reports a witness, it is a correct one.
\end{lemma}
\begin{proof}
The invariant of the algorithm is that $Z_i[k]$ counts the number of true witnesses for index $k$ during round $i$, where in the $i$-th round we only keep the indices $j$, with $j \in S$ (as $C[k](x)$ counts all possible witnesses and $R_S[k](x)$ the pseudo-witnesses).
This count, will always be non-negative, and will reduce after each successive round monotonically. After the counts are not greater than $c$, by a polylogarithmic number of repetitions, by Lemma \ref{lem: pos_prob} we can find all witnesses for the indices which are in $L'$ by leveraging that if an iteration leads to the count reaching $1$, we can apply Observation \ref{obs: fft_witness}, to find a true witness. 
\end{proof}
\begin{lemma}
All witnesses are found with high probability after $O(\log(n))$ rounds of the algorithm.
\end{lemma}
\begin{proof}
By the definition of a good set $S$, at most a fraction $\alpha$ of the indices $k$ will drop to $0$ (i.e. will not find a witness) in each round of the innermost loop going from 1 to t. As in the next iteration the sum of the entries in $Z_{i+1}$ is at most $3/4$-th of the entries, the inner most loop terminates in at most $t=\lceil 1+ 3\log_{4/3}(n) \rceil$ iterations. Thus, we find witnesses for a fraction of $1-t\alpha > 1/2$ many indices $k$. Repeating $O(\log(n))$ rounds of the algorithm boosts the probability to find witnesses for all indices $k$. 
\end{proof}
\begin{lemma}
The algorithm runs in expected time $\tilde{O}(n^{1.6}).$
\end{lemma}
\begin{proof}
Every FFT call runs in time $O(n^{1.6}\log n )$, the inner while loop runs at most $O(\log n )$ many times, as the sum of the entries in $Z_i$ drops in every iteration by at least 3/4-th. Finding all witnesses for the elements in $L'$ and also searching for a good set $S_{i+1}$ requires $\tilde{O}(n^{1.6})$ time.  We need $O(\log n )$ many rounds of the algorithm to boost the probability to find witnesses for all indices $k$, as argued previously.
\end{proof}
We remark that our algorithm can be derandomized by using $c$-wise $\epsilon$-dependent distributions, by following precisely the exact approach as Alon and Naor, see \cite{AlonN96}[Section 2.3].
Finally, we also remark that it is plausible that this idea can be extended for reporting witnesses for the $\tilde{O}(n^{1.5})$-algorithm from \cite{ChiDX022}, as it follows the same style of counting true witnesses.

\section{Additive Pareto sum approximation without witnesses} \label{sec:without_witnesses}
We call a set $\tilde{S}$ a \emph{weak} $\Delta$-Pareto sum approximation for sets $P,Q$ if it fulfills the following 3 criteria
\begin{itemize}
\item for all $p \in P, q\in Q$, there exists $ \tilde{s} \in \tilde{S}$  with $\tilde{s}\leq p+q+(\Delta,\Delta),$
\item for all $ \tilde{s} \in \tilde{S}$ there  exist $p \in P$ and $q \in Q$ such that $p+q \leq \tilde{s} \leq p+q+(\Delta,\Delta),$
\item $\tilde{S}$ is a Pareto set.
\end{itemize}
Thus in comparison to the usual definition only the second condition is relaxed. We describe Algorithm \ref{alg:weak_add_approx}, which computes a \emph{weak} $2t$-Pareto sum approximation without needing witnesses.
\begin{algorithm}
\caption{WeakApproximateParetoSum($P,Q,t$)}
\label{alg:weak_add_approx}
\begin{algorithmic}
\STATE{$\tilde{P} \gets \left \{ \left(\lceil \frac{p.x}{t} \rceil , \lceil \frac{p.y}{t} \rceil \right)\mid p\in P \right\}$}
\STATE{$\tilde{Q} \gets \left \{ \left(\lceil \frac{q.x}{t} \rceil , \lceil \frac{q.y}{t} \rceil \right)\mid q\in Q \right\}$}
\STATE{$S' \gets PS(\tilde{P},\tilde{Q})$}
\RETURN $\tilde{S} \gets$ Pareto front of $\{ts' \mid s \in S' \}$
\end{algorithmic}
\end{algorithm}
\begin{lemma}
Algorithm \ref{alg:weak_add_approx}, with input $(P,Q,t)$ leads to a \emph{weak} Pareto sum approximation, set $\tilde{S}$, "i.e.", fulfills the following 3 conditions: 
\begin{itemize}
\item for all $p \in P, q\in Q$, there exists $ \tilde{s} \in \tilde{S}$  with $\tilde{s}\leq p+q+2(t,t)$
\item for all $ \tilde{s} \in \tilde{S}$ there  exist $p \in P$ and $q \in Q$ such that $p+q \leq \tilde{s} \leq p+q+2(t,t)$
\item $\tilde{S}$ is a Pareto set.
\end{itemize}
\end{lemma}
\begin{proof}
\begin{itemize}
\item Let $p \in P$ and $q \in Q$, then there exists $s' \in S'$ with $s' \leq \tilde{p}+\tilde{q}$ (with $\tilde{p}=\left(\lceil \frac{p.x}{t} \rceil , \lceil \frac{p.y}{t} \rceil \right)$ and $ \tilde{q}=  \left(\lceil \frac{q.x}{t} \rceil , \lceil \frac{q.y}{t} \rceil \right)$) and 
$\tilde{s} =t s'.$ 
Then we have
\begin{align*}
\tilde{s}=ts' &\leq t\left(\lceil \frac{p.x}{t} \rceil , \lceil \frac{p.y}{t} \rceil \right) + t\left(\lceil \frac{q.x}{t} \rceil , \lceil \frac{q.y}{t} \rceil \right)\\
& \leq \left(t (\frac{p.x}{t} +1) , t(\frac{p.y}{t} +1) \right) + \left(t (\frac{q.x}{t} +1) , t(\frac{q.y}{t} +1) \right)\\
&\leq p+q +2(t,t).
\end{align*}
\item Let $\tilde{s} \in \tilde{S}$, then $\tilde{s}=ts'=t(\tilde{p} + \tilde{q})$ for some $ \tilde{p} \in \tilde{P}, \tilde{q} \in \tilde{Q}$.
Moreover $\tilde{p} = \left(\lceil \frac{p.x}{t} \rceil , \lceil \frac{p.y}{t} \rceil \right)$ and $\tilde{q} = \left(\lceil \frac{q.x}{t} \rceil , \lceil \frac{q.y}{t} \rceil \right) $ for some $q \in Q, p \in P$.
Observe that 
\begin{align*}
p+q &\leq \left(t\lceil \frac{p.x}{t} \rceil , t\lceil \frac{p.y}{t} \rceil \right) + \left(t\lceil \frac{q.x}{t} \rceil , t\lceil \frac{q.y}{t} \rceil \right)\\ 
& =t(\tilde{p}+ \tilde{q}) = \tilde{s}\\ 
&\leq \left(t (\frac{p.x}{t} +1) , t(\frac{p.y}{t} +1) \right) + \left(t (\frac{q.x}{t} +1) , t(\frac{q.y}{t} +1) \right)\\
& = p+q+2(t,t).
\end{align*}
\item Immediate by the last line of the algorithm.
\end{itemize}
\end{proof}
By making use of the $\tilde{O}(n^{1.5})$-algorithm of \cite{ChiDX022} and the above lemma, we get:
\begin{corollary}
Given sets $P,Q$ of $n$ points in $[0,W]^2$ and $\epsilon > 0$, we can compute a \emph{weak} $\epsilon W$-approximation of the Pareto sum in time $\ \tilde{\mathcal{O}} \left( n+ (1/\epsilon)^{1.5} \right)$. \label{thm: weak_epsW_runtime}
\end{corollary}
Furthermore, we remark that the lower bound Theorem \ref{thm: lower_bound} still holds for the weak Pareto sum approximation. This follows by the same proof and noticing that for any $\tilde{s}\in \tilde{S}$, Condition 2) establishes that there exist $p \in P, q\in Q$ such that $p+q \leq \tilde{s} \leq p+q+(\Delta,\Delta)$. By construction of $P,Q$ (since all coordinates are multiples of $2\Delta$), we can thus safely round down to the closest multiple of $2\Delta$ to obtain a corresponding better point $\tilde{s}'\le \tilde{s}$ that is a point in
 $P+Q$. Thus, we obtain a set $\tilde{S}'\subseteq P+Q$ satisfying the conditions of Definition~\ref{def:approx} as before, and the rest of the proof of Theorem~\ref{thm: lower_bound} can be used verbatim.
 
 %However by construction of the point sets $P,Q$ in the proof of Theorem~\ref{thm: lower_bound}, there is no point $p+q\in P+Q$ such that a $p'+q'$Implying $p+q \leq \tilde{s}\leq p+q$. Again it turns out that $\tilde{S}$ is the true Pareto sum of $P+Q$ and the proof follows.

\end{document}